\documentclass[twocolumn]{autart}    

\usepackage{graphicx}
\usepackage{amssymb,amsmath}
\usepackage{srcltx,pdfsync}
\usepackage{enumerate}
\usepackage{amsfonts}
\usepackage{epstopdf}
\usepackage{url}
\usepackage[dvipsnames,usenames]{color}
\usepackage[colorlinks, linkcolor=Blue, filecolor =Red, citecolor=RoyalPurple,]{hyperref}
\usepackage{multirow}
\usepackage{gensymb}
\usepackage{caption}
\usepackage{subcaption}

\newtheorem{proposition}{Proposition}
\newtheorem{definition}{Definition}

\newtheorem{prop-proof}{Proof of Proposition}
\newtheorem{thm-proof}{Proof of Theorem}

\makeatletter
\let\NAT@parse\undefined
\makeatother

\def\Ae{\ensuremath{A_\mathrm{e}}}
\def\matlab{MATLAB$^\copyright$}

\def\kmax{\ensuremath{k_\mathrm{max}}}

\def\algonamefull{Simultaneous Plant and Disturbance Identification
	through Regularization}
\def\algoname{SPDIR}

\def\dist{\ensuremath{q_\mathrm{int}}}
\def\qhvac{\ensuremath{q_\mathrm{hvac}}}
\def\solirr{\ensuremath{\eta^\mathrm{sol}}}

\def\mbf{\mathbf}

\def\eqdef{\ensuremath{:=}}

\def\R{\mathbb{R}}

\newcommand{\Lagr}{\mathcal{L}}
\newcommand{\G}{\mathcal{G}}

\def\zinv{\ensuremath{z^{-1}}}
\def\zinvv{\ensuremath{z^{-2}}}

\def\pb#1{\footnote{{\color{blue}{PB: #1}} }}

\def\version{Arxiv}

\newboolean{ArxivVersion}
\ifthenelse{\equal{\version}{Arxiv}}
{
	\setboolean{ArxivVersion}{true} 
}
{
	\setboolean{ArxivVersion}{false} 
}

\usepackage{xparse}

\ExplSyntaxOn
\cs_new:Nn \pb_prop_gset_bykeys:Nn
{
	\prop_clear:N \l__pb_temp_prop
	\keys_set:nn { pb/propbykey } { #2 }
	\prop_gset_eq:NN #1 \l__pb_temp_prop
}
\cs_generate_variant:Nn \pb_prop_gset_bykeys:Nn { c }

\keys_define:nn { pb/propbykey }
{
	unknown .code:n = \prop_put:Nxn \l__pb_temp_prop { \l_keys_key_tl } { #1 }
}

\cs_generate_variant:Nn \prop_put:Nnn { Nx }

\NewDocumentCommand{\setupcollaborator}{mm}
{
	\prop_new:c { g_collaborator_#1_prop }
	\pb_prop_gset_bykeys:cn { g_collaborator_#1_prop } { #2 }
}

\NewDocumentCommand{\selectcollaborator}{m}
{
	\prop_map_inline:cn { g_collaborator_#1_prop }
	{
		\tl_set:cn { ##1 } { ##2 }
	}
}
\ExplSyntaxOff

\setupcollaborator{PBmac}
{
	NAME=Prabir Mac , laptop or mac mini,
	DiCEbibPATH=/Users/pbarooah/Dropbox/DiCElab_bib,
	ACbibPATH=/Users/pbarooah/Dropbox/austin_bib,
	JBbibPATH=/Users/pbarooah/Dropbox/bibs,
	TZbibPATH = /Users/pbarooah/Dropbox/TTbib,
	figPATH=../figures,
}
\setupcollaborator{TZlinux}
{
	NAME=TZ at lab,
	DiCEbibPATH=/home/tingting/allBibFiles/DiCElab_bib,
	ACbibPATH=/home/tingting/allBibFiles/austin_bib,
	JBbibPATH=/home/tingting/allBibFiles/bibs,
	TZbibPATH = /home/tingting/allBibFiles/TTbib,
	figPATH= ./../figures,
}
\setupcollaborator{TZmac}
{
	NAME=TZ at her macbook,
	DiCEbibPATH=../../../bib/DiCElab_bib,
	ACbibPATH=../../../bib/austin_bib,
	JBbibPATH=../../../bib/bibs,
	TZbibPATH = ../../../bib/TTbib,
    figPATH=./,
}

\selectcollaborator{TZmac}

\begin{document}
	
	\begin{frontmatter}
		
		\title{Simultaneous identification of linear building dynamic model and disturbance using sparsity-promoting optimization\thanksref{footnoteinfo}} 
		
		\thanks[footnoteinfo]{This research is partially supported by NSF grants 1463316 and 1646229.}
		\author{Tingting Zeng}\ead{tingtingzeng@ufl.edu},
		\author{Jonathan Brooks}\ead{brooks666@ufl.edu},
		\author{Prabir Barooah$^*$}\ead{pbarooah@ufl.edu}
		
		\address{Mechanical and Aerospace Engineering, University of Florida,\\ Gainesville, Florida USA\\
			\vspace{10pt}	*Corresponding Author}
		

		\begin{abstract}    
		We propose a method that simultaneously identifies a linear time-invariant model of a building's temperature dynamics and a transformed version of the unmeasured disturbance affecting the building. Our method uses $\ell_1$-regularization to encourage the identified disturbance to be approximately sparse, which is motivated by the slowly-varying nature of occupancy that determines the disturbance. The proposed method involves solving a convex optimization problem that guarantees the identified black-box model possesses known properties of the plant, especially input-output stability and positive DC gains. These features enable one to use the method as part of a self-learning control system in which the model of the building is updated periodically without requiring human intervention. Results from the application of the  method on data from a simulated and real building are provided. 
		\end{abstract}
		
	\end{frontmatter}
	
	\section{Introduction}\label{sec:intro}
	A dynamic model of a building's temperature is useful for model-based fault detection and control of its HVAC (Heating Ventilation and Air Conditioning) system. There is a long history of such modeling efforts~\cite{Kramer:2012}. Due to the complexity of thermal dynamics, system identification from data is considered advantageous and there has been much work on it; see \cite{Kramer:2012,penman1990second,hu2016building} and references therein. A particular challenge for model identification is that temperature is affected by large, unknown disturbances, especially the cooling load induced by the occupants. The occupant-induced load refers to the heat gain directly due to the occupants' body heat and indirectly from lights and other equipments they use.  Another challenge comes from the need for automatic updates, especially for the use in model-based control. Due to changes in a building's properties over time, the model needs to be updated periodically with new data. 
	
	An identification method should also guarantee certain properties of the model so that it can be used as part of a self-learning control system without the need for a human expert to check the quality of the model. Most system identification methods for buildings ignore the unknown disturbances, but doing so can produce erroneous results. Only a few recent works have addressed the problem of unknown disturbances~\cite{kimcaiaribra:2016,coffman2018simultaneous}. None of the prior works however provide any guarantees on the properties of the identified model, such as stability. 
	
	In this paper we propose a method to estimate a linear dynamic model as well as a transformed version of the unknown disturbances from easily measurable input-output data. The method consists of solving a feasible and convex optimization problem, and the resulting model is guaranteed to possess properties that are known from physical insight into thermal dynamics of buildings, such as stability and positive DC gains of certain input output pairs. The proposed method, which we call \algoname\ (\algonamefull), is based on solving a constrained $\ell_1$-regularized least-squares problem. The $\ell_1$-penalty encourages the transformed disturbance to be a sparse signal. Use of the $\ell_1$-norm penalty to encourage sparse  solution is a widely used heuristic~\cite{tibshirani1996regression}. In our problem the motivation comes from the fact that the disturbance, which consists mostly of internal load due to occupants, is often slowly varying. For instance, large numbers of people enter and leave office buildings at approximately the same time. A slowly varying disturbance can be further approximated as piecewise-constant. We show that this feature makes the transformed disturbance an approximately sparse signal. The constraints ensure the identified model will have desirable properties. Evaluation of the method with simulation-generated data show that it can accurately identify the transfer function in the presence of large disturbances, even when the disturbance is not piecewise-constant. Evaluation with data from a real building are similarly promising, though accuracy is difficult to establish due to lack of a ground truth.

    A few works have partially addressed the challenge posed by  the presence of the unknown disturbance by using a specialized test building to measure the occupant-induced load~\cite{penman1990second,wang2006parameter}, or by collecting data from unoccupied times and setting the occupant-induced disturbance during that time to  0~\cite{fux2014ekf,hu2016building}. Work on model identification of building dynamics that handles occupant-induced heat gains in a principled manner, without requiring specially collected data, is limited. To the best of our knowledge, the only references that fall into this category are~\cite{kimcaiaribra:2016,coffman2018simultaneous}. There are many differences between these references and our work. We point out two key differences. One, 
	in contrast to the methods in~\cite{kimcaiaribra:2016,coffman2018simultaneous}, the proposed \algoname\ method can enforce properties of the system that are known from the physics of the thermal processes, in particular, stability and signs of DC gains for certain input-output pairs. For instance, an increase in outdoor temperature will lead to an increase in indoor temperature, but none of the prior methods guarantees that the identified model will predict this behavior. Second,  while the proposed \algoname\ method requires solving a feasible convex optimization problem, the methods in \cite{kimcaiaribra:2016,coffman2018simultaneous} require solving non-convex optimization problems. These two features of the proposed method enable it to be used as part of a self-learning control system without the need for a human expert in the loop. 

	\ifArxivVersion
	The solution to an $\ell_1$-regularized least-square problem is tunned by a regularization parameter $\lambda$. Although results from convex analysis in~\cite{tibshirani1996regression,ohlsson2012smoothed} tell us that there is a threshold $\lambda_{max}$ above which specific entries of the solution to such problem are identically zero, those works only determine $\lambda_{max}$ for formulations with no constraints, whereas linear inequality constraints are involved in our study. Considering the critical parameter value $\lambda_{max}$ provides very good starting point in finding a suitable value of $\lambda$. In this paper we determine $\lambda_{max}$ for an $\ell_1$-regularized least-square problem with linear inequality constraints along with a heuristic to choose $\lambda$.	
	\fi
	
	The article makes three additional contributions over the preliminary version~\cite{ZengSimultaneousHPB:2018}: (1) we determine the value of the critical regularization parameter $\lambda_{max}$ that is used in tuning the regularization parameter $\lambda$ (Proposition~\ref{prop:lambda_max}); (2) we provide evaluation of our method on data from a real building, and (3) 
	\ifArxivVersion we compare performance of the proposed method against the methods in~\cite{kimcaiaribra:2016,box2015time}. 
	\else
	we compare performance of the proposed method against the method in~\cite{kimcaiaribra:2016}. 
	\fi

	\ifArxivVersion
	The rest of this paper is organized as follows. Section~\ref{sec:problem_formu} formally describes the problem and establishes a few preliminaries. Section \ref{sec:method} describes the proposed algorithm. Section~\ref{sec:evaluation} provides evaluation results and Section~\ref{sec:conclusion} concludes the paper. In Section~\ref{sec:appendix}, constraints that enforce physical properties of the system are derived, and additional comparisons of the proposed method against the methods in~\cite{kimcaiaribra:2016,box2015time} are provided.
	\else
	The rest of this paper is organized as follows. Section~\ref{sec:problem_formu} formally describes the problem and establishes a few preliminaries. Section \ref{sec:method} describes the proposed algorithm. Due to the directive to reduce the paper to a brief paper format, some proofs of the technical results presented in Sections~\ref{sec:problem_formu} and ~\ref{sec:method} are omitted; they can be found in~\cite{ZengSimultaneousArXiV:2017}. Section~\ref{sec:evaluation} provides evaluation results and Section~\ref{sec:conclusion} concludes the paper. 
	\fi
	
	\section{Problem Formulation}\label{sec:problem_formu}
	The indoor zone temperature $T_z$ is affected by three \emph{known}	inputs: (1) the heat added to the zone by the HVAC system, \qhvac (kW), (2) the outside air temperature $T_{oa}$ (\degree C), (3) the solar irradiance \solirr (kW/m$^2$), and the \emph{unknown} disturbance $\dist$ (kW), which is the internal heat gain due to occupants, lights, and equipments used by the occupants. The only measurable output is the zone temperature $T_z$(\degree C). 
	
	Let $u(t) \eqdef[\qhvac(t), T_{oa}(t), \solirr(t)]^T \in \R^3$, $w(t) \eqdef \dist(t) \in \R$, and $y(t) \eqdef T_z(t) \in \R$. We start with the following second-order discrete-time transfer function model of  the system, with a sampling period $t_s$:
	\begin{align}\label{eq:dtmodel_1}
	\begin{aligned}
	y(\zinv) = &
	\frac{1}{D(\zinv)}\Big[\sum_{j=1}^3[\sum_{i=0}^2\alpha_{ij}z^{-i}]u_j(\zinv)\\
	& \quad + [\sum_{i=0}^2\beta_{i}z^{-i}]w(\zinv)\Big],
	\end{aligned}
	\end{align}
	where $D(\zinv) = 1-\theta_1\zinv-\theta_2\zinvv$, for some
	parameters $\theta_1,\theta_2$ and $\alpha_{ij},\beta_i$'s, and
	$u[k],w[k],y[k]$ are samples of the continuous-time signals $u(t),w(t),y(t)$. This model is a discrete-time version of a physics-based continuous time model that is described in Section~\ref{sec:RCmdl_const}. For future convenience, we rewrite~\eqref{eq:dtmodel_1} as
	\begin{align}\label{eq:dtmodel}
	\begin{aligned}
	y(\zinv) = &
	\frac{1}{D(\zinv)}\Big[K(\zinv)^Tu(\zinv)  + \bar{w}(\zinv) \Big], \\
	 \text{where}\quad& K(\zinv) := \begin{bmatrix}
	\theta_3\zinvv+\theta_4\zinv+\theta_5 \\
	\theta_6\zinvv+\theta_7\zinv+\theta_8 \\
	\theta_9\zinvv+\theta_{10}\zinv+\theta_{11} 
	\end{bmatrix},
	\end{aligned}
	\end{align}
	and $\bar{w}(\zinv)$ is the Z-transform of the \emph{transformed disturbance} signal $\bar{w}[k]$ defined as
	\begin{align}\label{eq:def-wbar}
	\bar{w}[k]&:= \beta_0w[k] + \beta_1w[k-1] + \beta_2w[k-2].
	\end{align}
	An inverse Z-transform on~\eqref{eq:dtmodel} yields a difference equation, which leads to:
	\begin{align} \label{eq:regression}
	y[k] = \phi[k]^T\theta, \quad k=3,\dots,\kmax,
	\end{align}
	where $\kmax$ is the number of samples, and $\theta^T:=[\theta_p^T,\bar{w}^T]$, in which $\theta_p=[\theta_1,\dots,\theta_{11}]^T \in \R^{11}$, $\bar{w}=[\bar{w}_3,\dots,\bar{w}_{\kmax}]^T \in \R^{\kmax-2}$ and  
	\begin{multline*}
	\phi[k]^T \eqdef 
	\Big[y[k-1],y[k-2],u_1[k-2],u_1[k-1],u_1[k],\\
	u_2[k-2],\dots,u_2[k],u_3[k-2],\dots,u_3[k], e_{k-2}^T\Big], 
	\end{multline*}
	where $e_k $ is the $k$-th canonical basis vector of $\R^{\kmax-2}$ in which the 1 appears in the $k^\text{th}$ place. Eq.~\eqref{eq:regression} can be expressed as:
	\begin{align}\label{eq:regression-vec}
	y = \Phi\theta, 
	\end{align}
	where $y \eqdef \left[y[3],\dots,y[\kmax]\right]^T \in \R^{\kmax-2}$ and
	\begin{align*}
	\begin{aligned}
	\Phi & \eqdef
	\begin{bmatrix}
	\phi[3]^T \\
	\dots\\
	\phi[\kmax]^T
	\end{bmatrix} \in \R^{\kmax-2 \times \kmax+9}.
	\end{aligned}
	\end{align*}
	The problem we seek to address is: \emph{given time traces of inputs and outputs, $\{u[k],y[k]\}_1^{\kmax}$, determine the unknown parameter vector $\theta_p \in \R^{11}$ and the unknown transformed disturbance vector $\bar{w}:=[\bar{w}_3,\dots,\bar{w}_{\kmax}]^T$, i.e., determine $\theta$}. 
	
	The matrix $\Phi$ is not full column-rank, so there will be an infinite number of solutions to \eqref{eq:regression-vec}. 
	\ifArxivVersion
	We also note that  $\Phi$ has the form
	\begin{align*}
	\Phi = \begin{bmatrix} \Psi_{(k_{max}-2) \times11}, & I_{(k_{max}-2) \times (k_{max}-2)} \end{bmatrix} .
	\end{align*}
	Since the number of samples is typically large, $\Psi$ is a tall matrix. Due to the dependency of $\Psi$ on (noisy) measurements of inputs and outputs, $\Psi$ is full column-rank except in case of degenerate data.
	\fi 
	We therefore will use physical insights to impose additional constraints on $\theta$ for the rest of this section.	
	
	\subsection{Parameter constraints from physical insights}\label{sec:para_const}
	The constraints described below are straightforward to derive, but involve - in a few cases - extensive algebra.
	\ifArxivVersion
	We therefore provide the proof in the Appendix.	
	\else
	We therefore omit the details here; they can be found in the expanded version~\cite{ZengSimultaneousArXiV:2017}.
	\fi
	
	\textbf{\textit{Stability}}
    The open loop dynamics of a building are bounded input bounded output (BIBO) stable; it will be a strange building indeed in which the temperature can become unbounded in response to bounded changes in the inputs. BIBO stability of the discrete-time model~\eqref{eq:dtmodel_1} is equivalent to:
    \ifArxivVersion
    \begin{align}
    	\begin{aligned}\label{eq:stable_const}
    	-\theta_2& < 1, 
    	\end{aligned}\\
    	\begin{aligned}\label{eq:stable_const2}
    	 \theta_2+\theta_1&< 1,
    	\end{aligned}\\    	
    	\begin{aligned}\label{eq:stable_const_redun}
    	\theta_2-\theta_1& < 1.
    	\end{aligned}
    	\end{align}
    \else
	\begin{align}\label{eq:stable_const}
	-\theta_2 < 1, \quad 
	\theta_2+\theta_1< 1,\quad	
	\theta_2-\theta_1 < 1.
	\end{align}
	\fi
	\textbf{\textit{Positive DC-gain}}
    In case of a real building, a steady state increase in the outdoor temperature will lead to a steady state increase in the indoor temperature, and the same pattern holds for each of the three inputs $\qhvac, T_{oa}, \eta^{sol}$. In other words, the corresponding DC gains must be positive. It can be shown that positive DC gains are equivalent to:
     \ifArxivVersion
	\begin{align}
	\begin{aligned}\label{eq:dc_num_const_1st}
	\theta_3+\theta_4+\theta_5&> 0,
	\end{aligned}\\
	\begin{aligned}\label{eq:dc_num_const_redun}
	\theta_6+\theta_7+\theta_8&> 0,
	\end{aligned}\\
	\begin{aligned}\label{eq:DCgain_const}
	\theta_9+\theta_{10}+\theta_{11}&> 0.
	\end{aligned}
	\end{align}
	\else        
	\begin{align}\label{eq:DCgain_const}
		\theta_{i}+\theta_{i+1}+\theta_{i+2}>0, \quad i \in \{3,6,9\}.
	\end{align}
	\fi
	
	\subsubsection{Physical insights from an RC network ODE model}\label{sec:RCmdl_const}
	\ifArxivVersion 
	Additional constraints can be imposed on $\theta$ if we use insights from a physics-based model. The physics-based model we use is a resistance-capacitance (RC) network model. An RC network is a common paradigm for modeling building thermal dynamics~\cite{Kramer:2012,fux2014ekf}. We will later assume that the discrete-time transfer function model~\eqref{eq:dtmodel_1} is obtained by discretizing a continuous-time RC network model, which helps us impose constraints on $\theta$.
	
	\begin{figure}[h]
		\includegraphics[width=0.95\linewidth]{\figPATH/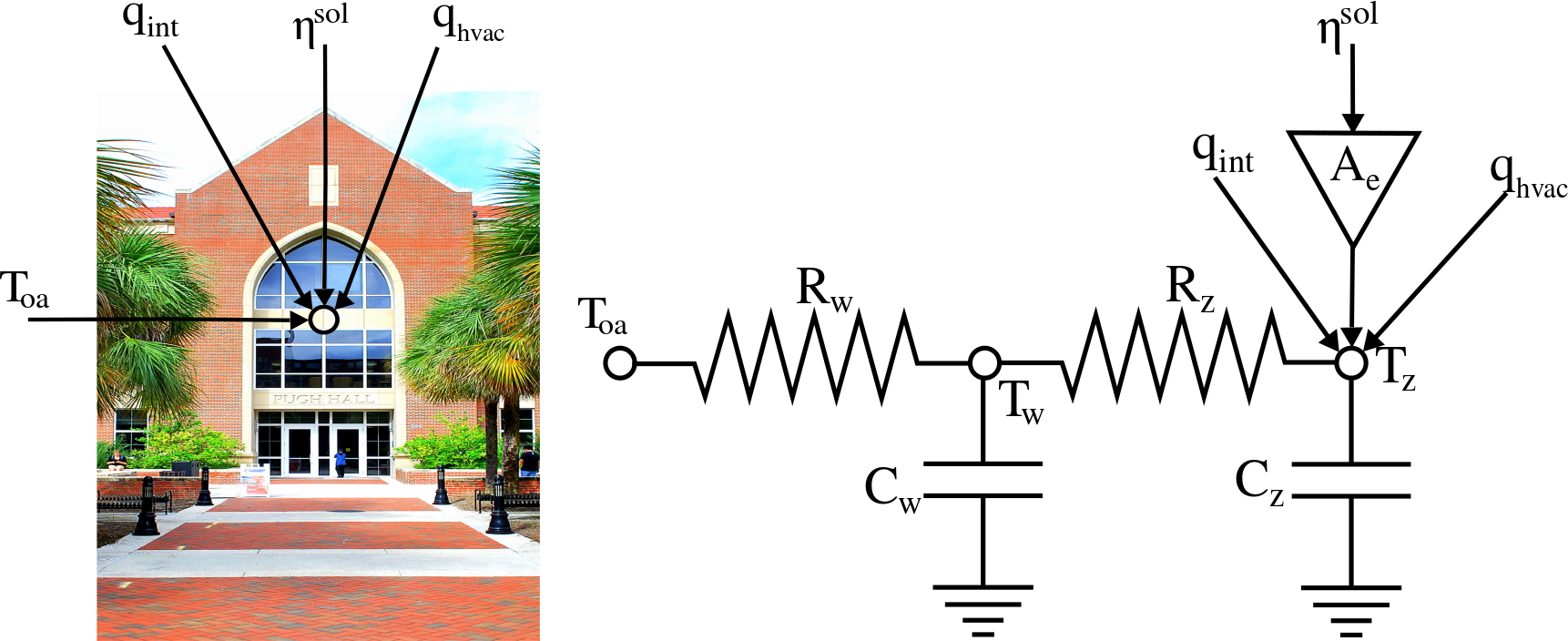}
		\centering
		\caption{A photograph of Pugh Hall and a schematic of the ``2R2C'' model.}
		\label{fig:RCschematic}
	\end{figure}
	Figure~\ref{fig:RCschematic} shows a building (left) and a corresponding 2nd-order resistance-capacitance (RC) network model (right). The ODE model of the RC-network model shown in the figure is:
	\else
	RC networks are widely used gray-box models for buildings~\cite{Kramer:2012,fux2014ekf}. Additional constraints can be imposed if we assume that the discrete-time transfer function model~\eqref{eq:dtmodel_1} is obtained by discretizing the following continuous-time resistance-capacitance (RC) network model: 
	\fi
	\begin{align} \label{eq:RC}
	\begin{aligned}
	C_z\dot{T}_z &= \frac{T_w-T_z}{R_z}+\qhvac+\Ae\solirr+\dist,\\ 
	C_w\dot{T}_w &= \frac{T_{oa}-T_w}{R_w}+\frac{T_z-T_w}{R_z},
	\end{aligned}
	\end{align}
	where $C_z, C_w, R_z, R_w$ are the thermal capacitances and resistances of the zone and wall, respectively, and \Ae\ is the effective area of the building for incident solar radiation. All five parameters are positive. Defining the state vector as $x \eqdef [T_z,
	T_w]^T \in \R^2$,  \eqref{eq:RC} can be written as
	\begin{align}\label{eq:ct_ss}
	\dot{x}&=Fx+Gu+Hw, & y &= Jx,
	\end{align}
	where $u$, $w$, and $y$ are defined at the beginning of Section~\ref{sec:problem_formu}, and $F\in \R^{2\times 2}, G \in \R^{2 \times 3}$, $H \in \R^{2
		\times 1}$ and $J \in \R^{1 \times 2}$ are appropriate matrices
	that are functions of the parameters $C_z, C_w, R_z, R_w,\Ae$. 
	\ifArxivVersion	
	Specifically,
	\begin{align*}
	F & = 
	\begin{bmatrix}
	-\frac{1}{C_zR_z} &      \frac{1}{C_zR_z} \\
	\frac{1}{C_wR_z} &     -\frac{1}{C_w}(\frac{1}{R_w}+\frac{1}{R_z})
	\end{bmatrix},\\
	G & = 
	\begin{bmatrix}
	\frac{1}{C_z}   &           0             & \frac{A_e}{C_z} \\
	0  &     \frac{1}{C_wR_w}    &     0
	\end{bmatrix}, \\
	H & = 
	\begin{bmatrix}
	\frac{1}{C_z}  \\  0 
	\end{bmatrix},\\
	J &=  
	\begin{bmatrix}
	1 & 0
	\end{bmatrix}.
	\end{align*} 
	\fi
	In Laplace domain,
	\begin{align} \label{eq:ct_tf}
	\begin{aligned}
	y(s) & = \frac{1}{D(s)} \Big[(s-f_{22}) \left(
	g_{11}u_1(s)+g_{13}u_3(s) \right)  \\
	& +f_{12}g_{22}u_2(s) + (s-f_{22})h_{11}w(s) \Big],
	\end{aligned}
	\end{align}
	where $f_{ij},g_{ij},h_{ij}$'s are the $i,j$-th entry of the matrices
	$F,G,H$ (respectively) in \eqref{eq:ct_ss}, and 
	\begin{align} 
	D(s) & = s^2 + d_1s + d_2, \; \text{with} \label{eq:def-D(s)}\\
	d_1 &=
	\frac{1}{C_zR_z}+\frac{1}{C_w}(\frac{1}{R_z}+\frac{1}{R_w}), &
	d_2 &= \frac{1}{C_zC_wR_zR_w}. \notag
	\end{align}
	\emph{We now assume that the discrete-time system~\eqref{eq:dtmodel_1} was obtained by discretizing the continuous-time system \eqref{eq:ct_tf} using Tustin transform.} It can be shown through straightforward calculations that the parameters of the discrete-time model -- the $\theta_i$'s -- are related
	to those of the continuous-time model \eqref{eq:ct_tf} as follows:
	\begin{align}
	\begin{aligned}
	&\theta_1 :=\frac{8-2d_2t_s^2}{D_0}, \;
	\theta_2 := -\frac{d_2t_s^2-2d_1t_s+4}{D_0}, \\
	&\begin{bmatrix}
	\theta_3 & \theta_9\\
	\theta_4 & \theta_{10}\\
	\theta_5 & \theta_{11}\\
	\end{bmatrix}  := \frac{t_s}{D_0}
	\begin{bmatrix}
	-2-f_{22}t_s \\-2f_{22}t_s \\ 2-f_{22}t_s
	\end{bmatrix}
	\begin{bmatrix}
	g_{11} & g_{13}
	\end{bmatrix}, \\
	& \begin{bmatrix}
	\theta_6 \\ \theta_7 \\ \theta_8
	\end{bmatrix}  := 
	\begin{bmatrix}
	1 \\ 2 \\ 1
	\end{bmatrix}
	\frac{f_{12}g_{22}t_s^2}{D_0},
	\end{aligned}\label{eq:theta_def}
	\end{align} 
	where $D_0=d_2t_s^2+2d_1t_s+4$. Similarly,
	\begin{align}
	[\beta_0,\beta_1,\beta_2] = \frac{t_s\big[(2+\epsilon_0), \;
		2\epsilon_0, \; (-2+\epsilon_0) \big]}{C_zD_0},   \label{eq:def-betas}\\
	\text{ where } \epsilon_0 = -f_{22}t_s = 
	\frac{t_s}{C_w}(\frac{1}{R_w}+\frac{1}{R_z}).   \label{eq:def-eps0}
	\end{align}	
	
	\textbf{\textit{Sign of parameters}}
	By using the positivity of the parameters $R_w,R_z,C_w,C_z,A_e$, the following holds:
	\begin{align}\label{eq:thetasigns}
	\begin{aligned}
	&\theta_i>0, \quad i \in \{1,4,5,6,7,8,10,11\}, \\
	&\theta_2<0, \quad \theta_3<0, \quad \theta_9<0,
	\end{aligned}
	\end{align}
	\ifArxivVersion
	of which the proof is provided in the Appendix.		
	\else
	whose proof is provided in~\cite{ZengSimultaneousArXiV:2017}.
	\fi
	
	\textbf{\textit{Sparse disturbance}}
	We need a few definitions to talk about \emph{approximately sparse}	vectors, and \emph{infrequently changing} vectors. 
	\begin{definition}
		\begin{enumerate}
		\item A vector $x\in \R^{n}$ is $(\epsilon,f)$-sparse if at most $f$ fraction of entries of $x$ are \emph{not} in $[-\epsilon,\epsilon]$.
		\item The \emph{change frequency} $c_f(x)$ of a vector $x\in\R^n$ is the fraction of entries that are distinct from their previous neighbor: $c_f(x) = \frac{1}{n-1}|\{k>1|x_k \neq x_{k-1}\}|$, where $|A|$ denotes the cardinality of the set $A$. We say a vector $x$ \emph{changes infrequently} if $c_f(x)\ll 1$.
		\end{enumerate}
	\end{definition}

	The following proposition shows that if the disturbance is slowly varying (e.g., if it is piecewise-constant), the transformed disturbance is approximately sparse.
	\begin{proposition}\label{prop:sparse}
		Suppose the disturbance $w[k]$ is uniformly bounded $|w[k]|\leq w_b$ in $k$, it changes infrequently with change frequency $c_f(\omega)$, and $\epsilon_0\ll 1$ where $\epsilon_0$ is defined in \eqref{eq:def-eps0}.  Then, $\bar{w}[k]$ is $(\bar{\epsilon},2c_f(w))$-sparse, where
		$\bar{\epsilon}=\frac{4}{C_zD_0}t_s w_b\epsilon_0$.
	\end{proposition}			
	\begin{prop-proof}
		It can be shown from \eqref{eq:def-wbar} and \eqref{eq:def-betas} that
		\begin{align*}
		\bar{w}[k]=&\frac{t_s}{C_zD_0}\big(2(w[k]-w[k-2])\\
		&\quad -\epsilon_0(w[k]+2w[k-1]+w[k-2])\big).
		\end{align*}
		Since $w$ is bounded, $\exists w_b\geq 0$ s.t. $w[k]\in
		[-w_b,w_b]$. Since $c_f(w) \ll 1$ from the hypothesis, for at least  $1-2c_f(w)$ fraction of  $k$'s,
		$w[k]-w[k-2]=0$, and for those $k$'s, 
		\begin{align*}
		\bar{w}[k]&=-\epsilon_0\frac{t_s}{C_zD_0}\big(w[k]+2w[k-1]+w[k-2]\big)\\
		&\in [\frac{-4\epsilon_0t_sw_b}{C_zD_0},\frac{4\epsilon_0t_sw_b}{C_zD_0}]=[-\bar{\epsilon},\bar{\epsilon}],
		\end{align*}
		which proves the result.     $\hfill \square$
	\end{prop-proof}

	Since the product $RC$ is large for large buildings, of the order of few hours~\cite{kimcaiaribra:2016}, $\epsilon_0$ is small for such buildings. In addition, both $\epsilon_0$ and  $\bar\epsilon$ can be made as small as possible by choosing $t_s$ sufficiently small. Therefore the assumption in Proposition~\ref{prop:sparse}, that $\epsilon_0$ is small, is not a strong one.
	
	\ifArxivVersion
	In order to ensure existence of a solution~\cite{luenberger1984linear}, the above constraints~\eqref{eq:stable_const}-\eqref{eq:DCgain_const}, and~\eqref{eq:thetasigns} are relaxed from strict inequalities to non-strict ones.
	
	\textbf{\textit{Redundancy of constraints}}
	After being relaxed into non-strict inequalities, constraints~\eqref{eq:stable_const}-\eqref{eq:DCgain_const}, and~\eqref{eq:thetasigns} can be compactly written as $\bar{g} =[\bar{g}_1^T, \bar{g}_2^T, \bar{g}_3^T, \bar{g}_4^T]^T\leq 0$, where 
	{\allowdisplaybreaks
     \begin{align*}
	\bar{g}_1(\theta_1,\theta_2) \eqdef&
	\begin{bmatrix}
	-1 & 0 \\ 0 & 1 \\ 0 & -1 \\ 1 & 1 \\ -1 & 1  
	\end{bmatrix}
	\begin{bmatrix}
	\theta_1 \\ \theta_2
	\end{bmatrix}
	+\begin{bmatrix}
	0 \\ 0 \\ -1 \\ -1 \\ -1
	\end{bmatrix}
	    \end{align*}
	\begin{align*}
	\bar{g}_2(\theta_3,\theta_4,\theta_5) \eqdef&
	\begin{bmatrix}
	1 & 0 & 0 \\ 0 & -1 & 0 \\ 0 & 0 & -1 \\ -1 & -1 & -1  
	\end{bmatrix}
	\begin{bmatrix}
	\theta_3 \\ \theta_4 \\ \theta_5
	\end{bmatrix} \\
	\bar{g}_3(\theta_6,\theta_7,\theta_8) \eqdef&
	\begin{bmatrix}
	-1 & 0 & 0 \\ 0 & -1 & 0 \\ 0 & 0 & -1 \\ -1 & -1 & -1  
	\end{bmatrix}
	\begin{bmatrix}
	\theta_6 \\ \theta_7 \\ \theta_8
	\end{bmatrix} \\
	\bar{g}_4(\theta_9,\theta_{10},\theta_{11}) \eqdef&
	\begin{bmatrix}
	1 & 0 & 0 \\ 0 & -1 & 0 \\ 0 & 0 & -1 \\ -1 & -1 & -1  
	\end{bmatrix}
	\begin{bmatrix}
	\theta_9 \\ \theta_{10} \\ \theta_{11}
	\end{bmatrix},
    \end{align*}}
	whose boundaries are shown in Figure~\ref{fig:feasi_reg}. 
	\begin{figure}[htb]
		\begin{subfigure}[htb]{1.05\linewidth}
			\centering
			\includegraphics[width=0.52\linewidth]{\figPATH/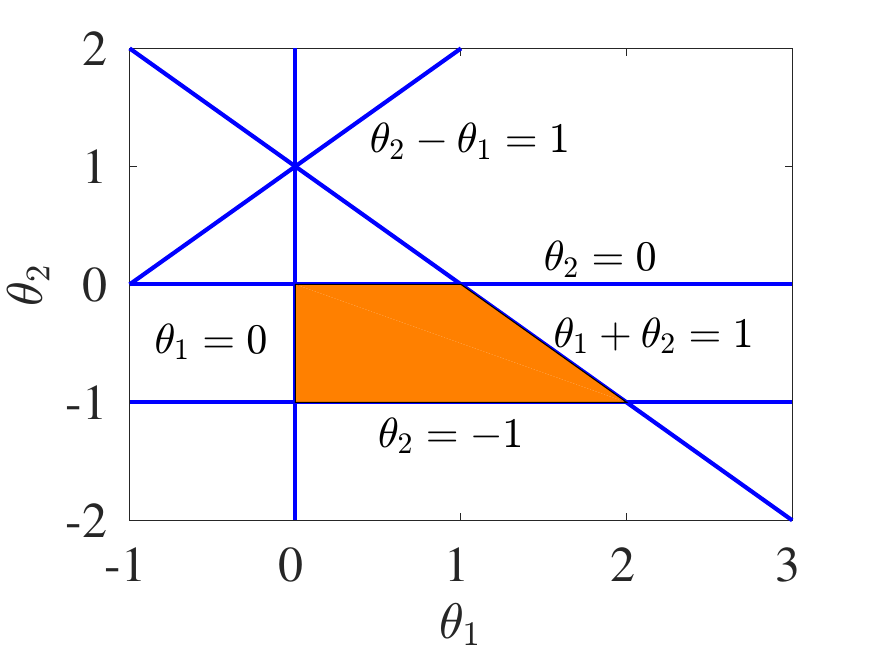}
			\includegraphics[width=0.47\linewidth]{\figPATH/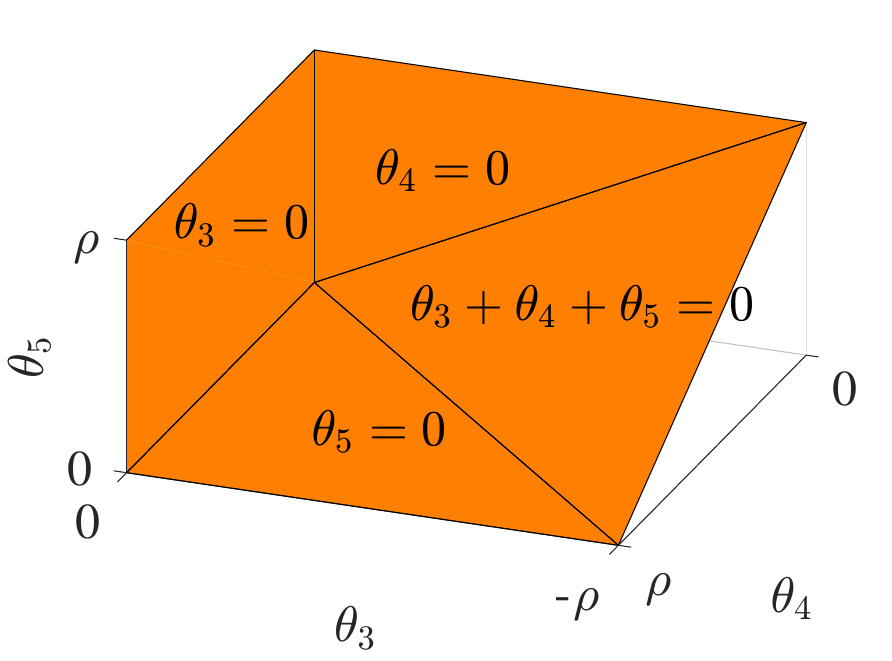}
			\caption{Boundaries of $\bar{g}_1(\theta_1, \theta_2)\leq 0$ are shown in blue (left), and a graph of feasible set $\G_1$ is shown in orange (left). Boundaries of $\bar{g}_2(\theta_3, \theta_4,\theta_5)\leq 0$ and $\G_2$ are shown in orange (right), where $\rho \rightarrow \infty$ (right).}
		\end{subfigure}
		\begin{subfigure}[htb]{1.05\linewidth}
			\centering
			\includegraphics[width=0.52\linewidth]{\figPATH/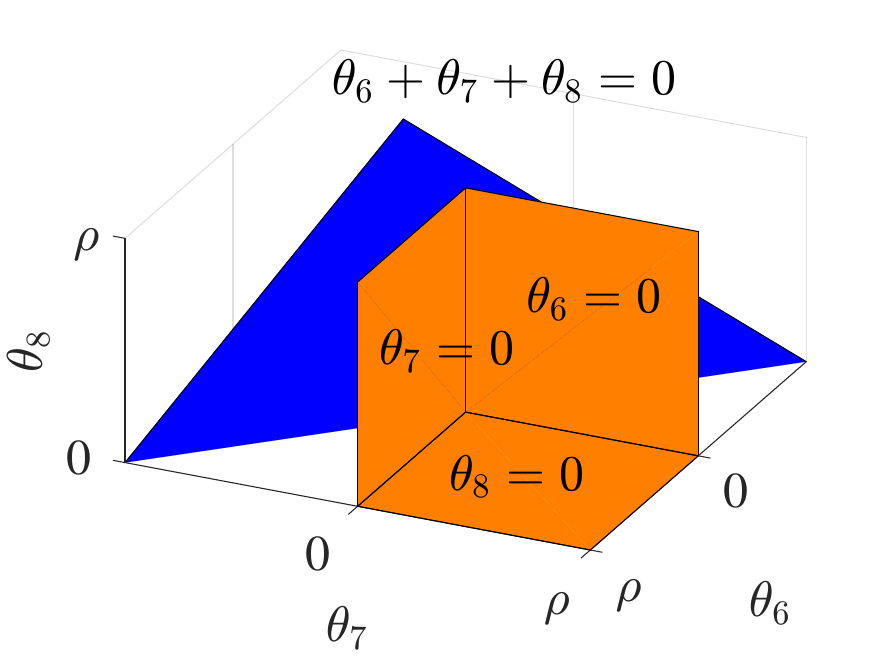}
			\includegraphics[width=0.47\linewidth]{\figPATH/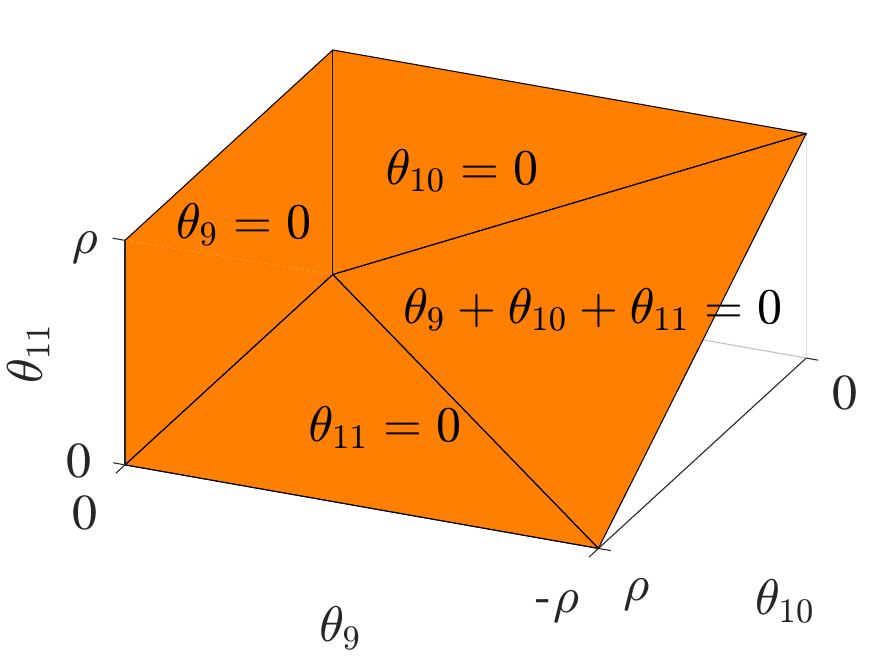}
			\caption{Boundaries of $\bar{g}_3(\theta_6, \theta_7,\theta_8)\leq 0 $ are shown in orange and blue (left), boundaries of $\G_3$ are shown in orange (left). Boundaries of $\bar{g}_4(\theta_9, \theta_{10},\theta_{11})\leq 0$ and $\G_4$ are shown in orange (right). Here $\rho \rightarrow \infty$.}
		\end{subfigure}
		\caption{Feasible sets $\G_k$'s are non-empty and convex.}
		\label{fig:feasi_reg}
	\end{figure}
	Denote the feasible sets for $\bar{g}_1,\bar{g}_2,\bar{g}_3,\bar{g}_4\leq 0$ as
	\begin{align*}
	\begin{aligned}
	\G_1 &= \{(\theta_1,\theta_2)|\bar{g}_1(\theta_1,\theta_2)\leq 0\}\\
	\G_2 &= \{(\theta_3,\theta_4,\theta_5)|\bar{g}_2(\theta_3,\theta_4,\theta_5)\leq 0\}\\
	\G_3 &= \{(\theta_6,\theta_7,\theta_8)|\bar{g}_3(\theta_6,\theta_7,\theta_8)\leq 0\}\\
	\G_4 &= \{(\theta_9,\theta_{10},\theta_{11})|\bar{g}_4(\theta_9,\theta_{10},\theta_{11})\leq 0\}, \\
	\end{aligned}
	\end{align*}
	respectively. The set $\G_1$ and boundaries of $\G_i,$ $i=2, 3, 4$ are shown in Figure~\ref{fig:feasi_reg}.
	
	Noticing from Figure~\ref{fig:feasi_reg}(a) (left) and Figure~\ref{fig:feasi_reg}(b) (left), the last inequality from $\bar{g}_1\leq 0$, i.e., constraint~\eqref{eq:stable_const_redun}, and the last one from $\bar{g}_3\leq 0$, i.e., \eqref{eq:dc_num_const_redun}, are redundant. Mathematically,
	\begin{align*}
	\begin{aligned}
	\bigcap\limits_{i=1}^{5}\{(\theta_1,\theta_2)|\bar{g}_{1,R_i}\leq 0\}&=\bigcap\limits_{i=1}^{4}\{(\theta_1,\theta_2)|\bar{g}_{1,R_i}\leq 0\}\\
	\bigcap\limits_{i=1}^{4}\{(\theta_6,\theta_7,\theta_8)|\bar{g}_{3,R_i}\leq 0\}&=\bigcap\limits_{i=1}^{3}\{(\theta_6,\theta_7,\theta_8)|\bar{g}_{3,R_i}\leq 0\},\\
	\end{aligned}
	\end{align*}
	where $\bar{g}_{1,R_i}: \R^2\rightarrow \R$ and $\bar{g}_{3,R_i}: \R^3\rightarrow \R$ is the i-th entry of the function $\bar{g}_k$, where $k=1,3$ respectively (imagining $\bar{g}_k$ as a column vector). Therefore constraints~\eqref{eq:stable_const_redun} and \eqref{eq:dc_num_const_redun} can be removed without changing the feasible sets. The remaining, linearly independent constraints can be written as
	\begin{align*}
		G_c^u\theta_p+g_c\leq 0, \quad G_c^u:\R^{11} \to \R^{15}
	\end{align*}
	where $G_c^u$ is a full column-rank block diagonal matrix,
	\begin{align*}
	G_c^u &=  diag\Bigg( 
	\left[	\begin{smallmatrix}
	-1 & 0 \\ 0 & 1 \\ 0 & -1 \\ 1 & 1  
	\end{smallmatrix}\right]
		\left[\begin{smallmatrix}
	1 & 0 & 0 \\ 0 & -1 & 0 \\ 0 & 0 & -1 \\ -1 & -1 & -1  
	\end{smallmatrix}\right]
		\left[\begin{smallmatrix}
	-1 & 0 & 0 \\ 0 & -1 & 0 \\ 0 & 0 & -1
	\end{smallmatrix}\right]
		\left[\begin{smallmatrix}
	1 & 0 & 0 \\ 0 & -1 & 0 \\ 0 & 0 & -1 \\ -1 & -1 & -1  
	\end{smallmatrix}\right]\Bigg) \notag \\
	g_c &= 
	\begin{bmatrix}
		0 & 0 & -1 & -1 &0_{1\times 11}
	\end{bmatrix}^T. 
	\end{align*}	
	For future convenience, we write the constraints in the equivalent form:
	\begin{align*}
	\begin{aligned}
	    G_c\theta+g_c\leq 0, \; \text{where} \;
	    G_c = \begin{bmatrix}
	    G_c^u, 
	    & 0_{15\times k_{max}-2}\\
	    \end{bmatrix},\quad
	\end{aligned}
	\end{align*}
	where the inequality is entry-wise.	
	\else
	Some redundant constraints from~\eqref{eq:stable_const}-\eqref{eq:DCgain_const} and~\eqref{eq:thetasigns} can be removed without changing the feasible set; see~\cite{ZengSimultaneousArXiV:2017} for the details. The remaining, linearly independent constraints	can be written as $G_c\theta+g_c\leq 0$, where
	\begin{align*}
	\begin{aligned}
		g_c &= 
		\begin{bmatrix}
			0 & 0 & -1 & -1 &0_{1\times 11}
		\end{bmatrix}^T,
		G_c =
		\begin{bmatrix}
			G_c^u, 0_{15\times k_{max}-2}\\
		\end{bmatrix},\\
		G_c^u &=  diag\Bigg( 
		\left[	\begin{smallmatrix}
			-1 & 0 \\ 0 & 1 \\ 0 & -1 \\ 1 & 1  
		\end{smallmatrix}\right]
		\left[\begin{smallmatrix}
			1 & 0 & 0 \\ 0 & -1 & 0 \\ 0 & 0 & -1 \\ -1 & -1 & -1  
		\end{smallmatrix}\right]
		\left[\begin{smallmatrix}
			-1 & 0 & 0 \\ 0 & -1 & 0 \\ 0 & 0 & -1
		\end{smallmatrix}\right]
		\left[\begin{smallmatrix}
			1 & 0 & 0 \\ 0 & -1 & 0 \\ 0 & 0 & -1 \\ -1 & -1 & -1  
		\end{smallmatrix}\right]\Bigg).  \\
	\end{aligned} 
	\end{align*}
	\fi

	\section{Proposed \algoname\ Algorithm}\label{sec:method}
	Since we expect $w$ to be piecewise-constant and infrequently changing, $\bar{w}$ should be approximately sparse. Let $S\eqdef [0_{\kmax-2\times11} | \;I_{\kmax-2}]$ so that $S\theta=\bar{w}$. We therefore seek a solution to $\mbf{y}=\Phi\theta$ so that $S\theta$ is sparse, by posing the following optimization problem:
	\begin{align}\label{eq:opt} 
	\begin{aligned}
	\hat{\theta} & = \arg \min_{\theta }  \frac{1}{2} \|y -  \Phi \theta\|_2^2 +\lambda \|S\theta\|_1 \\
	& \quad \text{ s. t. } G_c\theta+g_c\leq 0,
	\end{aligned}
	\end{align}
	where $\lambda \geq 0$ is a user-defined weighting factor. The $\ell_1$-norm penalty is to encourage sparsity of the solution; see the discussion in Section \ref{sec:intro}. 
	\ifArxivVersion
	The problem \eqref{eq:opt} is an extension of the so-called ``generalized lasso" problem~\cite{ali2018generalized}:
	\begin{equation*}
	\hat{\theta}  = \arg \min_{\theta }  \frac{1}{2} \|y -  \Phi \theta\|_2^2 +\lambda \|S\theta\|_1,
	\end{equation*}
	with the extension being the addition of the linear inequality constraint.  We therefore call problem \eqref{eq:opt} the ``linearly constrained generalized lasso problem", or lcg-lasso for short. 
	\else
	Problem~\eqref{eq:opt} is called the ``linearly constrained generalized lasso problem", or lcg-lasso for short.
	\fi 
	The estimated plant parameters $\hat{\theta}_p$ and estimated transformed disturbance $\hat{\bar{w}}$ can be recovered from $\hat\theta$ since $\theta^T  = [\theta_p^T, \bar{w}^T]$.
	
	\ifArxivVersion
	The next result establishes a few properties of the optimization problem \eqref{eq:opt}.
	\else
	The next result establishes a few properties of the optimization problem \eqref{eq:opt}, whose proof can be found in~\cite{ZengSimultaneousArXiV:2017}.
	\fi	
  	We call a point $\theta$ \emph{physically meaningful} if none of the three SISO transfer functions in~\eqref{eq:dtmodel} is identically zero. 
	\begin{proposition}\label{prop:feasible_convex_regular}
		The optimization problem \eqref{eq:opt} is feasible, convex, and every physically meaningful feasible $\theta$ is a regular point of the constraints. 
	\end{proposition}
	\ifArxivVersion
	\begin{prop-proof}
	The feasible set for the constraint $G_c\theta+g_c\leq0$ is
		\begin{align*}
		\G \eqdef \G_1 \times \G_2 \times \G_3 \times \G_4,
		\end{align*}
		where $\times$ denotes the Cartesian product. Because $\G_k$'s are non-empty and convex, $\G$ is also non-empty and convex. The objective function is convex since it is a sum of two convex functions. Therefore the optimization problem \eqref{eq:opt} is feasible and convex.

		Notice that the origins in the sets $\G_2$, $\G_3$, and $\G_4$ are not physically meaningful as defined above, and these are the only non-meaningful points. Hence, at any physically meaningful feasible point, each $\bar{g}_k$ will have no more than two active constraints. It can be verified by inspection (see Figure \ref{fig:feasi_reg}) that the gradients of these active constraints are linearly independent. Therefore, every physically meaningful feasible point is a regular point of the constraints.
		$\hfill \square$
	\end{prop-proof}
	\fi

	\subsection{Regularization Parameter Selection}\label{sec:regu_select}
	The selection of $\lambda$ determines the solution to lcg-lasso problem~\eqref{eq:opt}. At one extreme, $\lambda = 0$ will lead to a least-squares solution to \eqref{eq:opt} that will suffer from over-fitting. A larger $\lambda $ will make the resulting $S\theta$ sparser. We therefore propose a heuristic to select $\lambda$ by searching in a range $[0, \lambda_{\max}]$. The following proposition provides both the value of $\lambda_{\max}$ and the rationale for stopping the search at that value. 	
	\begin{proposition}\label{prop:lambda_max}
	Every solution $\hat{\theta}$ to \eqref{eq:opt} satisfies $S\hat \theta  = 0 =\hat{\bar{w}}$ if and only if $\lambda > \lambda_{max}\eqdef ||y||_\infty$. 
	\end{proposition}
	\begin{prop-proof}
	Since all inequalities are affine, and $\theta=0$ is  feasible, a weaker form of Slater's condition is satisfied which means strong duality holds~\cite[eq.~(5.27)]{boyd2004convex}. Let $\beta \eqdef \Phi\theta, \quad 	\chi\eqdef S\theta, \quad z\eqdef G_c\theta$.  The augmented Lagrangian function of~\eqref{eq:opt} is:
	\begin{align*}
	   \Lagr(\theta,& z,\chi,\beta;\gamma,\zeta,\mu,\eta) 
	    = \frac{1}{2} \|y - \beta\|_2^2 
	         +\lambda \|\chi \|_1 
	         + \gamma^T(z+g_c) \notag \\
	    & + \mu^T(\chi -S\theta)
	         + \eta^T(\beta-\Phi\theta)  
	         + \zeta^T(z-G_c\theta),
	\end{align*}
	where $\gamma \geq 0$. 	The dual function is 
	\begin{align*}
	\begin{aligned}
	&g(\gamma,\zeta,\mu,\eta) = \inf_{\theta,z,\chi,\beta}\Lagr\\
	& \quad = \inf_\theta-(\eta^T\Phi+\mu^TS+\zeta^TG_c)\theta
	        + \inf_z(\zeta^T+\gamma^T)z  \\
	& \quad +\inf_\chi (\lambda||\chi||_1 +\mu^T\chi)
        	+\inf_\beta(\frac{1}{2} \|y - \beta\|_2^2 +\eta^T\beta )	 
	        + \gamma^Tg_c.
	\end{aligned}
	\end{align*}
	Since a linear function is bounded below only when it is identically zero, thus
	\begin{align*}
	\begin{aligned}
	&\inf_\theta-(\eta^T\Phi+\mu^TS+\zeta^TG_c)\theta=
	    \begin{cases}
	     0 & \Phi^T\eta=-S^T\mu-G_c^T\zeta\\
	     -\infty & \text{otherwise}
	    \end{cases},\\
	&\inf_z(\zeta^T+\gamma^T)z=
	    \begin{cases}
	    0 & \zeta+\gamma=0, \gamma\geq 0\\
	    -\infty & \text{otherwise}
	    \end{cases},\\    
	&\inf_\chi (\lambda||\chi||_1 +\mu^T\chi) 
	=\sum_{k=1}^{k_{max}-2}\inf_{\chi_k} (\lambda|\chi_k| +\mu_k\chi_k)\\
	&\hspace{80pt}=\begin{cases}
	0 & ||\mu||_\infty \leq \lambda\\
	-\infty & \text{otherwise}
	\end{cases}.
	\end{aligned}
	\end{align*}	
	The corresponding minimizers for $||\mu||_\infty\leq \lambda$ satisfy:
	\begin{align}\label{eq:sol_inf_chi}
	\begin{cases}
	\text{if}\; \mu_k = -\lambda, \;\; \hat{\chi}_k = \text{any non-negative number} \\
	\text{if}\; |\mu_k| < \lambda, \;\; \hat{\chi}_k = 0\\
	\text{if}\; \mu_k = \lambda, \;\;\hat{\chi}_k = \text{any non-positive number}\\
	\end{cases}.
	\end{align}
	Finally the infimum over $\beta$ is
	\begin{align*}
	\inf_\beta(\frac{1}{2} \|y - \beta\|_2^2 +\eta^T\beta ) = 
	    \frac{1}{2} \|y \|_2^2 - \frac{1}{2} \|y - \eta\|_2^2,
	\end{align*}
	which is derived by setting $\frac{\partial\Lagr}{\partial\beta}=0$ and substituting the resulting minimizer $\beta = y-\eta$. 
	Therefore the dual function can be simplified as
   	\begin{align}\label{eq:dual_func}
   	&g(\gamma,\mu,\eta,\zeta) =
   	\begin{cases}
   	\frac{1}{2} \|y \|_2^2 - \frac{1}{2} \|y - \eta\|_2^2 + \gamma^Tg_c & C1\\
   	-\infty & \text{o/w}
   	\end{cases},
   	\end{align}
   	where $C1$ stands for the following:
   	\begin{align}\label{eq:opt_dual_cond}
   	C1: \quad       \begin{cases}
   	\Phi^T\eta=-S^T\mu-G_c^T\zeta\\
   	\zeta+\gamma=0, \; \gamma\geq 0\\
   	||\mu||_\infty \leq \lambda.
   	\end{cases}
   	\end{align}     
	The dual variables $\gamma,\mu,\eta,\zeta$ are dual feasible because \eqref{eq:opt_dual_cond} has a trivial solution. The first equation from \eqref{eq:opt_dual_cond} has the form:
	    \begin{align}
	    \begin{bmatrix}
	    \Psi^T_{11\times (k_{max}-2)} \\ I_{k_{max}-2}
	    \end{bmatrix}
	    \eta & = -
	    \begin{bmatrix}
	    0_{11\times (k_{max}-2)} \\ I_{k_{max}-2}
	    \end{bmatrix}
	    \mu -
	    \begin{bmatrix}
	    (G_c^u)^T_{11\times 15}\\
	    0_{(k_{max}-2)\times 15}\\
	    \end{bmatrix}\zeta, 	\notag \\
	    \implies &
	    \begin{aligned}
	    &\Psi^T\eta = -(G_c^u)^T\zeta, \\
	    &\eta  = -\mu.
	    \end{aligned}\label{eq:phi_S_G1_struc2}
	    \end{align}    
	which has infinite number of solutions $(\eta,\mu,\zeta)$ since $\Psi^T$ and $(G_c^u)^T$ both have full row rank. Eliminating $\eta$ and $\zeta $ from \eqref{eq:dual_func} using \eqref{eq:opt_dual_cond}-\eqref{eq:phi_S_G1_struc2}, the dual problem is
	\begin{align}\label{eq:dual_prob}
	\begin{split}
	(\hat{\gamma},\hat{\mu}) & = \max_{\gamma,\mu}  \frac{1}{2} \|y \|_2^2 - \frac{1}{2} \|y + \mu\|_2^2 + \gamma^Tg_c\\
	& \text{ s. t. } 
	     -\Psi^T\mu = (G_c^u)^T\gamma, \gamma \geq 0,\\
	& \qquad     ||\mu||_\infty \leq \lambda.
	\end{split}
	\end{align}
	For a given $\lambda\geq 0$, two scenarios arise when solving \eqref{eq:dual_prob}. 

	Scenario 1:  \emph{$\lambda \leq \|y\|_\infty$:} In this scenario, the $k$-th entry of any solution $\hat{\mu}$ to \eqref{eq:dual_prob} will satisfy $		|\hat{\mu}_k| = \min(\lambda,|y_k|)$ and there is at least one entry  that satisfies $|\hat{\mu}_k| = \lambda$. The corresponding solution $\hat{\chi}_k$ is non-unique according to \eqref{eq:sol_inf_chi}. Hence $\hat{\chi}$ is non-unique. 

	Scenario 2:  \emph{$\lambda > \|y\|_\infty$:} In this case the solution to \eqref{eq:dual_prob} satisfies $\hat{\mu}=-y$, and therefore, $\|\hat{\mu}\|_\infty = \|y\|_\infty < \lambda$. From \eqref{eq:sol_inf_chi}, we have that $\hat{\chi}=0$.  Since $\chi = S\theta = \bar{w}$, the result is proved. $\hfill \square$
	\end{prop-proof}	

	\textbf{\textit{Heuristic for selecting $\lambda$:}}
	The heuristic we propose to choose $\lambda$ is based on the L-curve method, and uses the result from the previous proposition. First, plot both the solution norm $\| S\theta\|_1$ and residual norm $\|y -  \Phi \theta\|_2$ against $\lambda$ by repeatedly solving Problem~\eqref{eq:opt}  for variocus $\lambda $ in $ [0,\; \lambda_{max}]$, where $\lambda_{max}$ is defined in Proposition~\ref{prop:lambda_max}. An illustration of these two plots is shown in Figure~\ref{fig:L1-L2-curve}. Second, identify a value $\lambda_1$ so that the solution norm is smaller than a user-defined threshold for $\lambda>\lambda_1$, and then identify $\lambda_2$ so that the residual norm is smaller than a user-defined threshold for $\lambda<\lambda_2$. If $\lambda_2>\lambda_1$, choose $\lambda$ to be $\lambda_1$. If not, pick another threshold, and continue until this condition is met. Figure~\ref{fig:L1-L2-curve} shows an example of having these curves both lie in picture.
	\begin{figure}[htb]
		\includegraphics[width=0.85\linewidth]{\figPATH/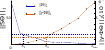}
		\centering
		\caption{Illustration of regularization parameter selection}		
		\label{fig:L1-L2-curve}
		\vspace{-0.2 cm}
	\end{figure}
	
	\section{Evaluation of Proposed \algoname\ algorithm}\label{sec:evaluation}
	Numerical implementation of the proposed method is performed by using the \texttt{cvx} package for solving convex problems in \matlab~\cite{cvx}.

    Two experiments are conducted in order to test the proposed method \algoname, one using simulation data and the other using real building data collected from Pugh Hall, a commercial building in the University of Florida campus. The method proposed in~\cite{kimcaiaribra:2016} is also implemented as a comparison, which is referred to as the LD (Lumped Disturbance) method. We remark here the LD method is non-convex, and the results from the LD method presented here are obtained with  a multi-start approach with random initial guesses.
    
    Simulation data is generated by simulating the continuous-time RC model~\eqref{eq:RC}. The parameters of the model were taken from \cite[Table 1]{coffman2018simultaneous}, which uses a model of the same structure.
	\ifArxivVersion
	Four scenarios are tested:
	\begin{enumerate}
		\item \emph{OL-PW:} Open-loop with piecewise-constant disturbance;
		\item \emph{OL-NPW:} Open-loop with not piecewise-constant disturbance;
		\item \emph{CL-PW:} Closed-loop with piecewise-constant disturbance;
		\item \emph{CL-NPW:} Closed-loop with not piecewise-constant disturbance;
	\end{enumerate}	
	The algorithm is expected to perform well in the OL-PW scenario since the disturbance satisfies the piecewise-constant assumption the method is designed for, and identification with open-loop data is generally easier than with closed loop data~\cite{ljungbook:1999}. The CL-NPW scenario is the most relevant in practice, but it is likely to be the most challenging for the method. We therefore focus on discussing the results from the CL-NPW scenario in the sequel. 

	In the two open-loop scenarios, the input component $\qhvac$ is	somewhat arbitrarily chosen, while in the closed-loop scenarios, \qhvac\ is decided by a PI-controller that tries to maintain the zone temperature at a setpoint $T^\mathrm{ref}$. To have exciting input to aid in identification, the setpoint $T^\mathrm{ref}$ is chosen to be a PRBS sequence~\cite{ljungbook:1999}. To ensure that occupant comfort is not compromised, the setpoint is constrained to	lie within $22$\degree C and $27$\degree C. The input components, ambient temperature from \url{weatherunderground.com}, and solar irradiance data from NSRDB: \url{https://nsrdb.nrel.gov/}, both for Gainesville, FL, are used in all four scenarios. The disturbance signal $\dist$ is chosen by scaling CO$_2$ measurements from Pugh Hall. We assume the scheduled occupancy is correlated to the CO$_2$ level. The training data are shown in Fig.~\ref{fig:week_input}. 
	\begin{figure}[tb]
		\centering
		\includegraphics[width=1\linewidth]{\figPATH/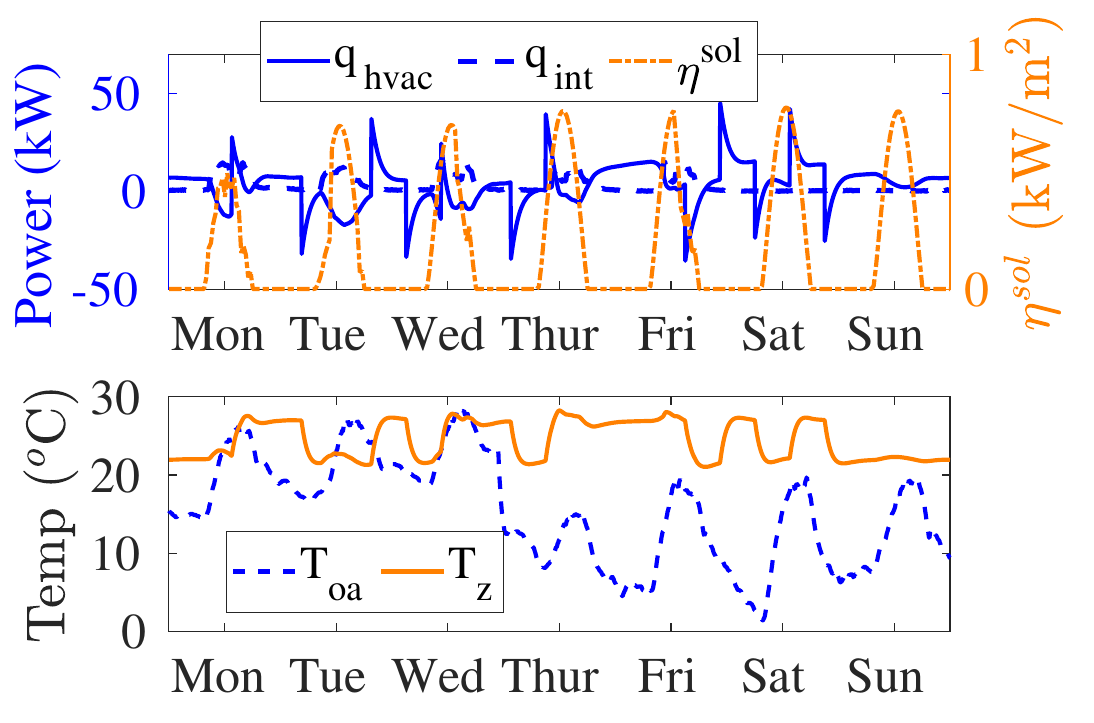}
		\caption{Training data from simulated building. The data  $\solirr, T_{oa},\dist$ shown here are used in all four scenarios; 
			$\qhvac, T_z$ shown here are for the CL-NPW scenario.}
		\label{fig:week_input}
	\end{figure}	
	
    For the real building, measurements of $\qhvac$ and $T_z$ are collected from a large zone (an auditorium) in Pugh Hall. See Fig.~\ref{fig:week_input_bldg}. The location of the room from which measurements are collected is shown in Fig.~\ref{fig:location_audi}. The ambient temperature and solar irradiance data, collected from the same online source at another week, are used.
	\begin{figure}[tb]
		\centering
		\includegraphics[width=1\linewidth]{\figPATH/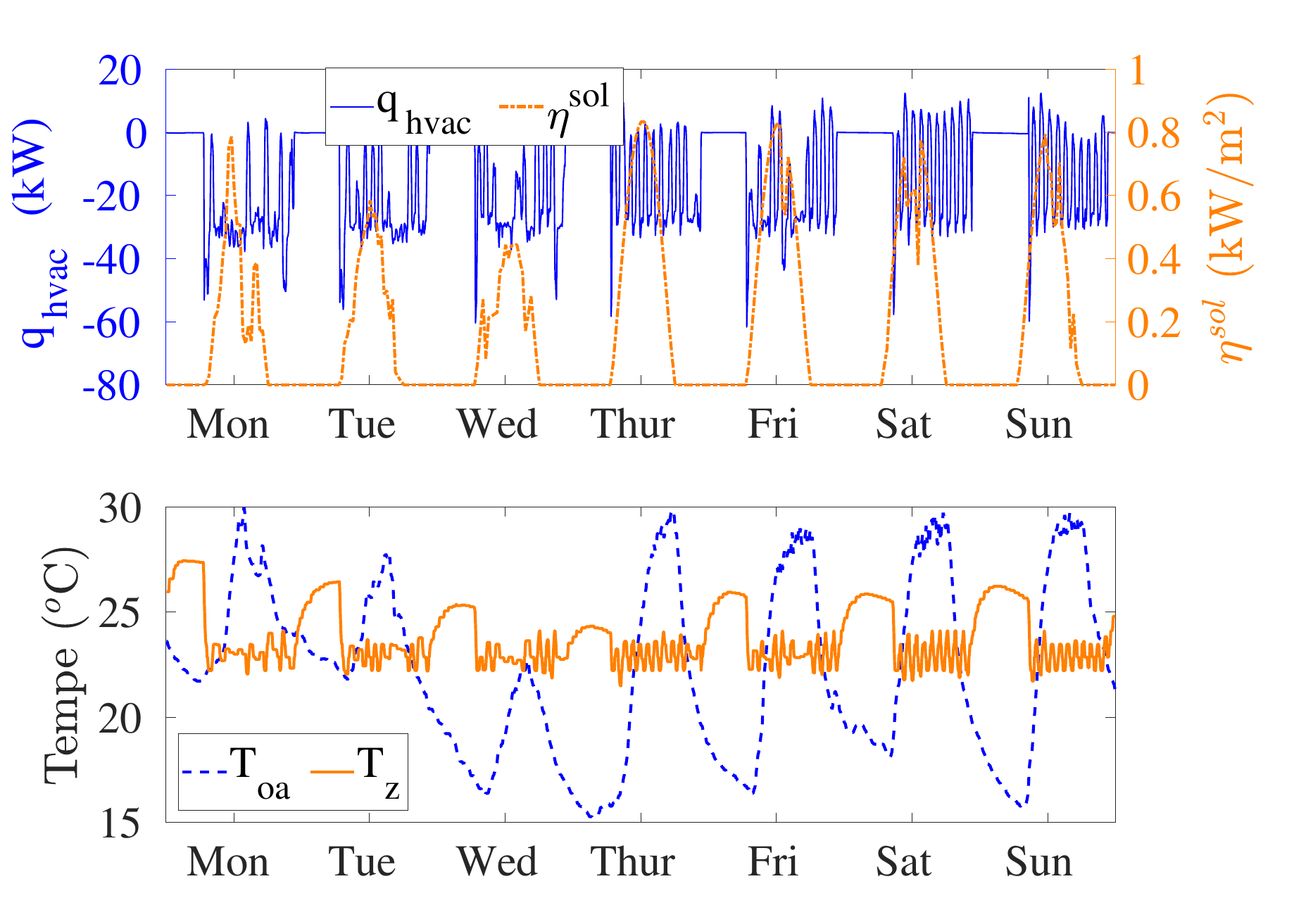}
		\caption{Training data from real building.}
		\label{fig:week_input_bldg}
	\end{figure} 	
	\else
	More details can be found in~\cite{ZengSimultaneousArXiV:2017}.
	
	We remark here that the simulation experiment is designed to put the proposed method to test: (i) data are collected from a closed-loop simulation, and (ii) the disturbance signal is not piecewise-constant.
	
	For the real building, measurements of $\qhvac$ and $T_z$ are collected from a large zone (an auditorium) in Pugh Hall; see Fig.~\ref{fig:location_audi}. The ambient temperature and solar irradiance data, collected from the same online source in~\cite{ZengSimultaneousArXiV:2017} at another week, are used.
	\fi     
	\begin{figure}[htb]
		\centering
		\includegraphics[width=0.95\linewidth]{\figPATH/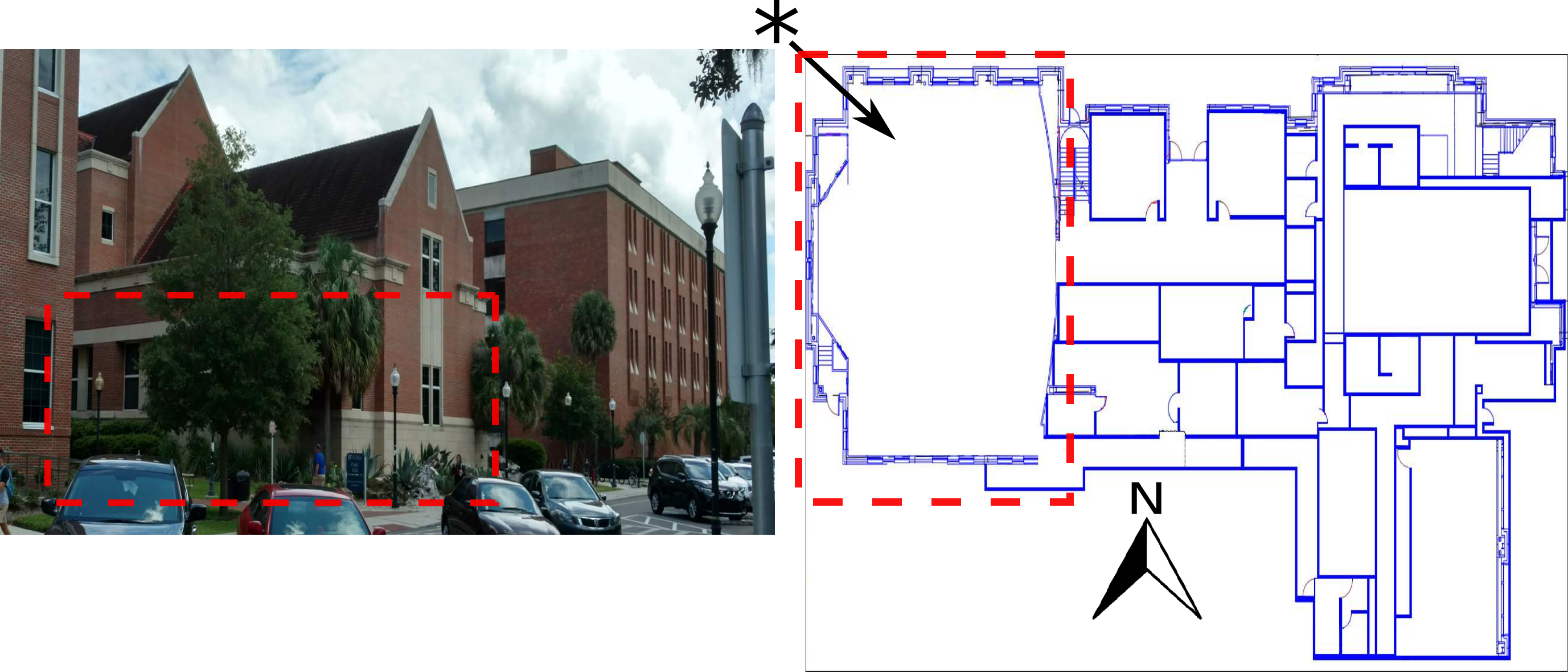}
		\caption{Pugh hall photograph (left) and floor plan (right), with the zone from which building data are collected shown enclosed in dashed lines. The ``$\ast$" denotes the location from where the photograph was taken and the arrow denotes direction of the camera.}
		\label{fig:location_audi}
	\end{figure}
	
	\subsection{Algorithm evaluation with simulation data}\label{Sec:results-plantid}
	\textbf{\textit{Parameters}}
	Table~\ref{tab:theta-compare} shows the true values of the plant
  	parameters, $\theta_p$, and the	corresponding estimation errors (in percentage) \ifArxivVersion from CL-NPW scenario\fi. 
  	First, the two parameters, $\theta_1, \theta_2$, that determine the characteristic equation are estimated highly accurately. Second, there is more error in the estimate of numerators. While some are more accurate than others, the numerator coefficients corresponding to the input $\solirr$ has the most error. A possible reason for this high error is the lack of richness in the $\solirr$ data. Parameter estimates by the proposed method are slightly more accurate than those by the LD method.
		\begin{table}[htb]		
		\caption{Plant parameters and errors in their estimates.}	
		\begin{center}
			\begin{tabular}{c|c|c|c|c}
				\multicolumn{2}{c|}{\multirow{2}{*}{$\theta_p$}} & \multicolumn{2}{c|}{$\frac{\hat{\theta}_p-\theta_p}{\theta_p}\%$} & 
				\multirow{2}{*}{input}\\ 
				\multicolumn{2}{c|}{}                  &  LD   &  SPDIR   &  \\ \hline
				$\theta_1   $ & $ 1.97565 $ & $-0.061 $ & $ -0.042 $ &  \\
				$\theta_2   $ & $-0.97573 $ & $-0.13  $ & $ -0.085 $ &  \\ \hline
				$\theta_3   $ & $-4.35\times 10^{-3} $ & $ 94.3  $ & $  8.0   $ &  \\
				$\theta_4   $ & $ 5.21\times 10^{-5} $ & $ 16262 $ & $ 108.2  $ & \qhvac \\
				$\theta_5   $ & $ 4.41\times 10^{-3} $ & $-100.0 $ & $   6.4  $ &  \\ \hline
				$\theta_6   $ & $ 1.86\times 10^{-5} $ & $ 67.9  $ & $ 48.9   $ &  \\
				$\theta_7   $ & $ 3.72\times 10^{-5} $ & $ -15.3 $ & $  -22.3 $ & $T_{oa}$ \\
				$\theta_8   $ & $ 1.86\times 10^{-5} $ & $-100.0 $ & $ 68.4   $ &  \\ \hline
				$\theta_9   $ & $-3.05\times 10^{-2} $ & $ 430.5 $ & $232.1   $ &  \\
				$\theta_{10}$ & $ 3.65\times 10^{-4} $ & $ 44577 $ & $19324   $ & \solirr \\
				$\theta_{11}$ & $ 3.08\times 10^{-2} $ & $-100.0 $ & $ 2.9    $ & \\ \hline
			\end{tabular}			
			\label{tab:theta-compare}
		\end{center}
	\end{table}	

	\ifArxivVersion
	\begin{figure*}[htb]	
		\begin{minipage}[t]{0.5\linewidth}
			\includegraphics[width=0.95\linewidth]{\figPATH/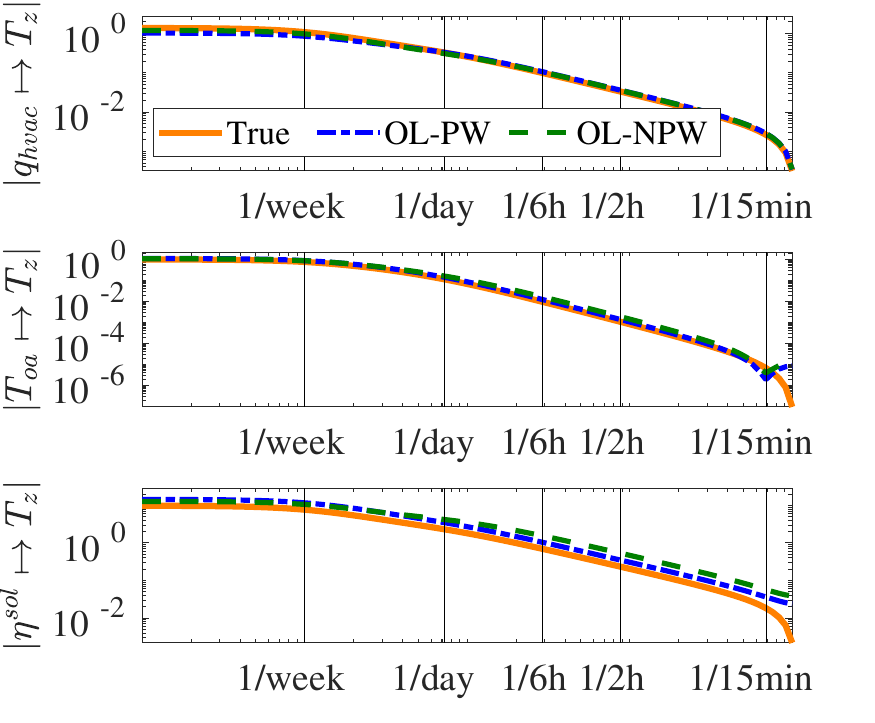}
			\centering
		\end{minipage}\hspace{0.01\linewidth}
		\begin{minipage}[t]{0.5\linewidth}	
			\includegraphics[width=0.95\linewidth]{\figPATH/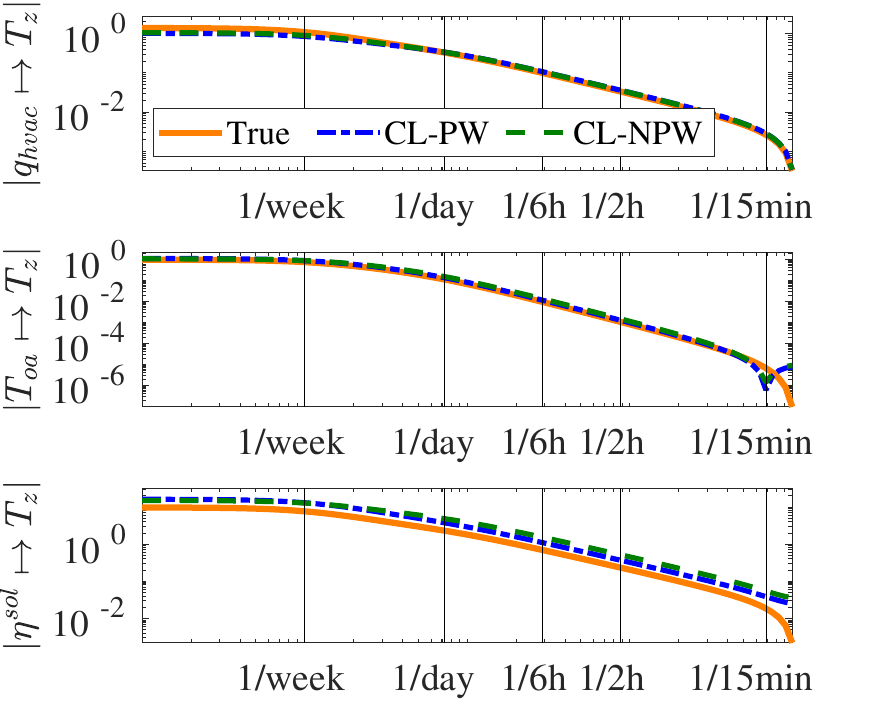}
			\centering
		\end{minipage}\hspace{0.01\linewidth}
		\caption{Algorithm evaluation on simulation data: Bode magnitude plots of the true and identified systems for 4 scenarios.}
		\label{fig:bode_prop_4cases}		
	\end{figure*}	
	\begin{figure}[htb]	
		\includegraphics[width=0.95\linewidth]{\figPATH/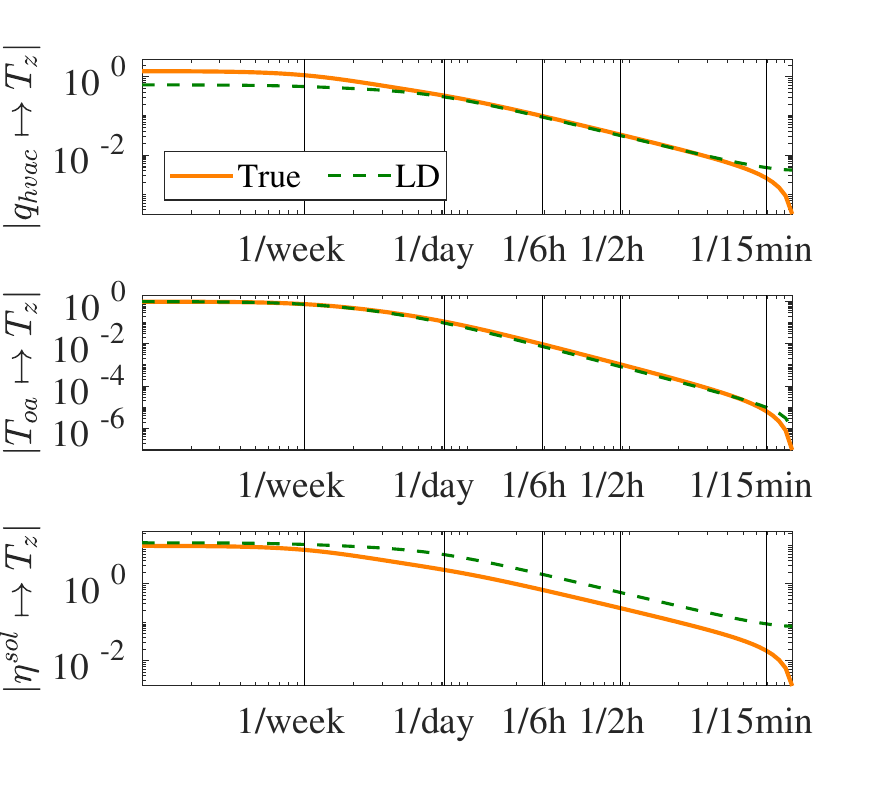}
		\centering
		\caption{LD method application on simulation data: Bode magnitude plots of the true and identified systems.}		
		\label{fig:bode_sim_LD}
	\end{figure} 
	\else
	\begin{figure}[htb]
		\includegraphics[width=0.99\linewidth]{\figPATH/bode_sim_prop_LD.pdf}
		\centering
		\caption{Algorithm evaluation on simulation data: Bode magnitude plots of the true and identified systems.}		
		\label{fig:bode_jnl_cmp}
	\end{figure}
	\fi
	   	 
	\textbf{\textit{Frequency response}}
	For prediction accuracy, frequency response is more important than individual parameters. 
	\ifArxivVersion
	Figure~\ref{fig:bode_prop_4cases} compares the frequency response of the identified plants with their true values for the two open-loop and closed-loop scenarios, respectively.
	\else
	Fig.~\ref{fig:bode_jnl_cmp} shows the Bode plots of the true and identified models. \fi	
	For the transfer function from input $\qhvac$ to output $T_z$, the maximum absolute error in the estimated frequency response is:
	\begin{align*}
	\max_{\omega}\frac{|\hat{G}_{\qhvac \rightarrow T_z}(j\omega) - G_{\qhvac \rightarrow T_z}(j\omega)|}{|G_{\qhvac \rightarrow T_z}(j\omega) |} = 0.24
	\end{align*}
	\ifArxivVersion and occurs at $ \omega = 1/(10$ weeks) for the CL-PW scenario. The maximum errors for the transfer functions from $T_{oa}$ and $\eta^{sol}$ to $T_z$ occur at the Nyquist frequency. Frequency responses of identified models from the LD methods for the CL-PW scenario are similar; see Figure~\ref{fig:bode_sim_LD}.
	\else
	and occurs at $ \omega = 1/(10$ weeks). The maximum errors for the transfer functions from $T_{oa}$ and $\eta^{sol}$ to $T_z$ occur at the Nyquist frequency. Frequency responses of identified models from the proposed \algoname\ and the LD methods are similar.
	\fi 
	   
    \textbf{\textit{Disturbance}}
    \ifArxivVersion
    The estimated transformed disturbance, $\hat{\bar w}$, for all four scenarios are shown in Fig.~\ref{fig:Wbar_cmp}. The estimates are quite accurate when the true values are large, but less accurate otherwise. However, the estimates capture the trend of the true values,
    \begin{figure}[htb]
    	\includegraphics[width=1.1\linewidth]{\figPATH/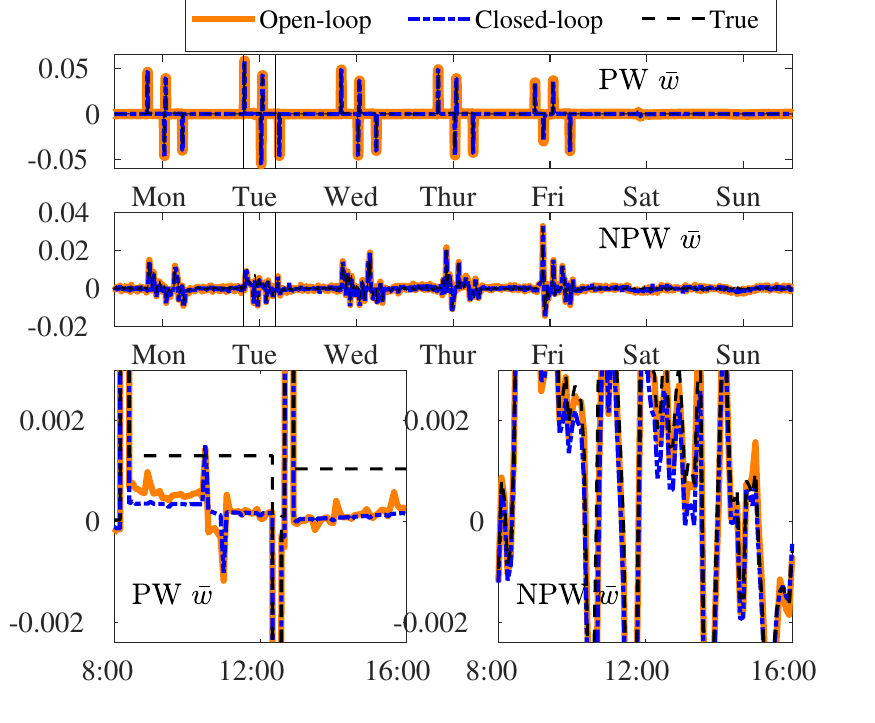}
    	\centering
    	\caption{Algorithm evaluation on simulation data: comparison of identified and actual transformed disturbance. Bottom two plots are zoomed version on Tuesday of the top two plots.}		
    	\label{fig:Wbar_cmp}
    \end{figure}	
    \else    	
	The estimated transformed disturbance, $\hat{\bar w}$, is shown in Fig.~\ref{fig:Wbar_cmp}. The estimates capture the trend of the true values,
    \begin{figure}[htb]
    	\includegraphics[width=0.99\linewidth]{\figPATH/wbar_sim_CLNPW.pdf}
    	\centering
    	\caption{Algorithm evaluation on simulation data: comparison of identified and actual transformed disturbance. Bottom plot is zoomed version on Thursday of the top plot.}		
    	\label{fig:Wbar_cmp}
    \end{figure}
    \fi
    even when the true disturbance is not piecewise-constant, in which case the transformed disturbance may be neither approximately sparse nor infrequently changing. Since the LD method identifies an output disturbance while the proposed method identifies an input disturbance, the disturbance estimates are not comparable.
        
	\textbf{\textit{Zone temperature prediction}}
	The plant identified with data from one week is used to predict temperatures in another week. The disturbance data is the same between the training and validation data sets but the input $u$ and output $y$ data sets are distinct. 
	\ifArxivVersion
	\begin{figure*}[htb]
		\includegraphics[width=0.9\linewidth]{\figPATH/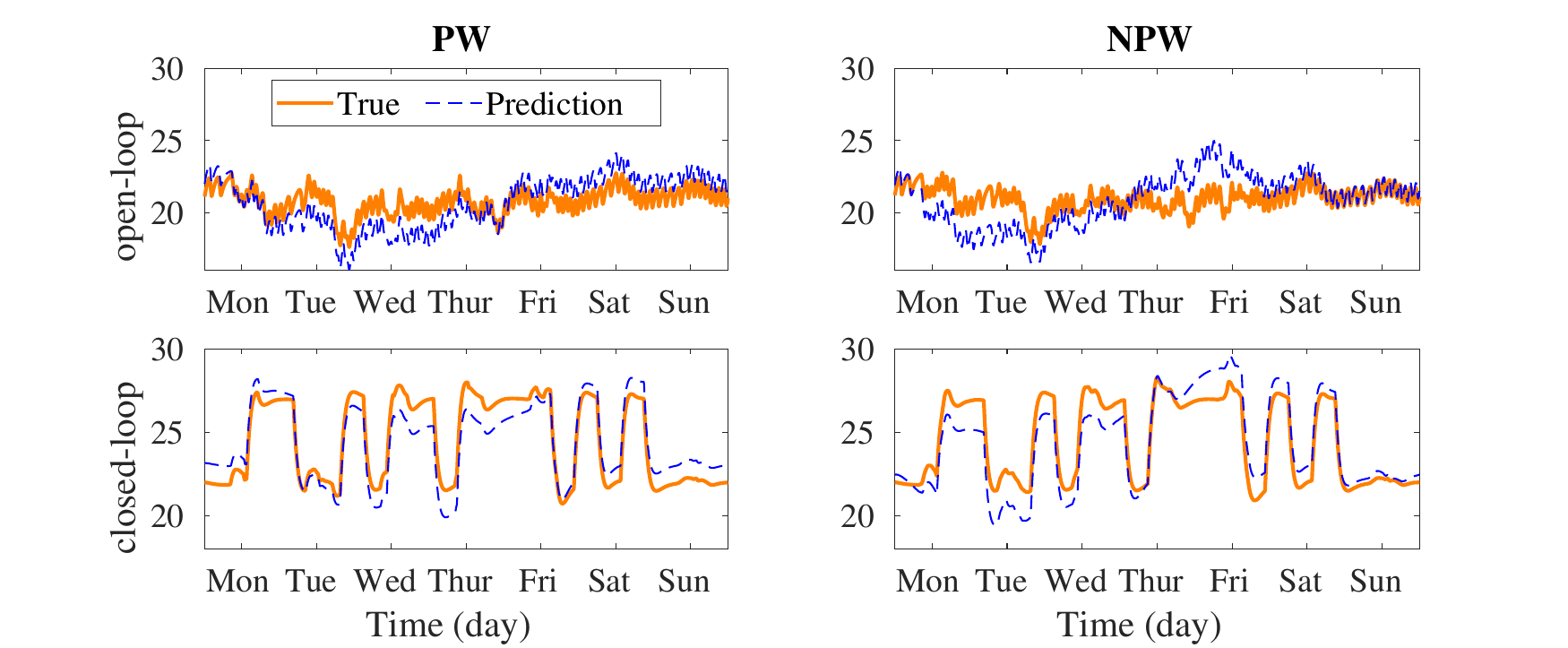}
		\centering
		\caption{Algorithm evaluation on simulation data: predicted and actual zone temperature (validation dataset).}
		\label{fig:Tr_sim_prop4case}
	\end{figure*}
	\begin{figure}[htb]
	\includegraphics[width=0.99\linewidth]{\figPATH/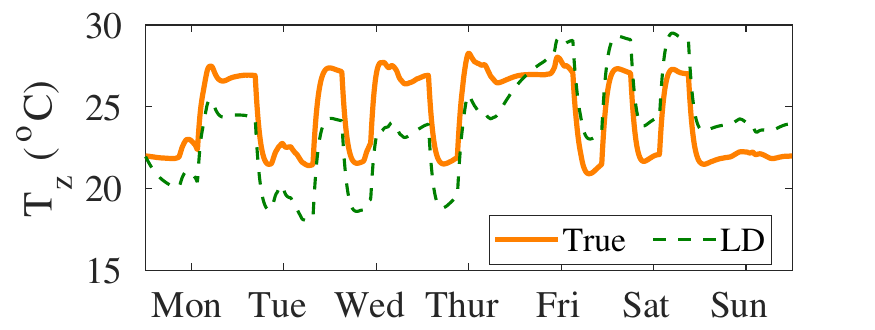}
	\centering
	\caption{LD methods application on simulation data: predicted and actual zone temperature (validation dataset).}
	\label{fig:Tr_sim_LD}
	\end{figure}	
	The rms value of the prediction error of zone temperature is 1.2 \degree C for the proposed \algoname\ method in all 4 scenarios, see Figure~\ref{fig:Tr_sim_prop4case}. As we can see from the figure, the error is more pronounced in certain days of the week. The rms error is 1.1 \degree C for the LD method, see Figure~\ref{fig:Tr_sim_LD}.
	\else
	The rms value of the prediction error of zone temperature is 1.2 \degree C for the proposed \algoname\ method, and 1.1 \degree C for the LD method.
	\fi

    The LD method performs comparably to the proposed method in these tests because the initial guess for the non-convex optimization problem in the LD method was chosen carefully. When initial guesses are not chosen carefully, the proposed method outperforms the LD method. 
    \ifArxivVersion
    Details of the comparison are available in Sec.~\ref{sec:append_cmp_LD}.
    \else
    Details of the comparison are available in~\cite{ZengSimultaneousArXiV:2017}; they are omitted here due to lack of space.
    \fi
    
	\subsection{Algorithm evaluation using building data}\label{Sec:results-bldg-id}
	Evaluation with data from a real building is challenging since there is no ground truth to compare with.
	
	\textbf{\textit{Frequency response}}
	Fig.~\ref{fig:Bode_cmp_bldg} shows the Bode plots of the identified model for the real building. Notice that the Bode plots generated using both simulation data and building data are similar. Since the simulation model's parameters are taken from ~\cite{coffman2018simultaneous}, which were obtained by applying the system identification method proposed in that reference to the  data from the same building (Pugh Hall's auditorium), this similarity provides confidence in the results. Frequency responses of the model identified by the LD method are similar in lower frequencies but less so in higher frequencies.
	\begin{figure}[htb]
		\includegraphics[width=0.9\linewidth]{\figPATH/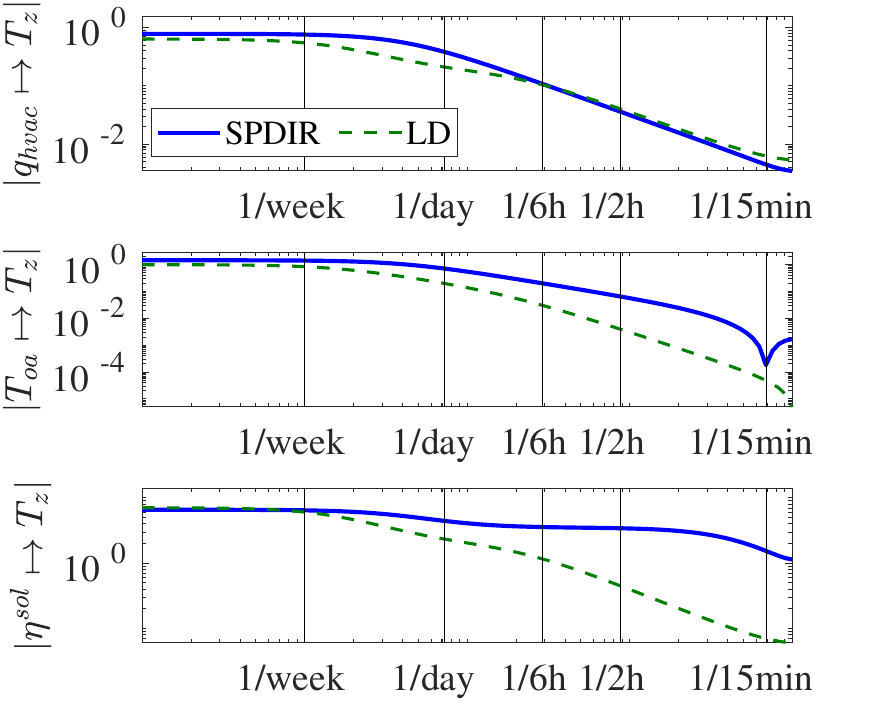}
		\centering
		\caption{Algorithm evaluation on building data: Bode magnitude plots of identified systems.}
		\label{fig:Bode_cmp_bldg}
	\end{figure}
	\begin{figure}[htb]
		\includegraphics[width=0.99\linewidth]{\figPATH/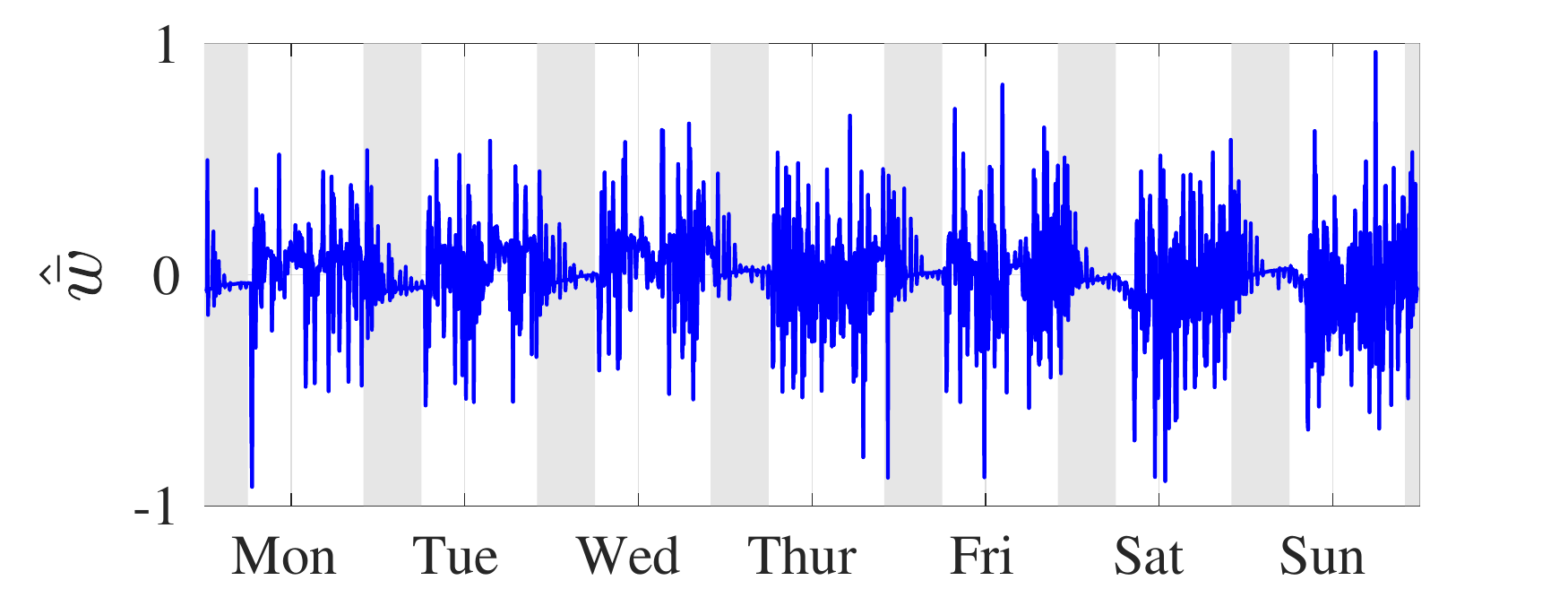}
		\centering
		\caption{Algorithm evaluation on building data: identified transformed disturbance. Night time shaded in gray.}		
		\label{fig:Wbar_bldg}
	\end{figure}
	
    \textbf{\textit{Disturbance}}
    The estimated transformed disturbance $\hat{\bar w}$ is shown in Fig.~\ref{fig:Wbar_bldg}. The entries corresponding to nighttime are small in magnitude. This is consistent with what we expect: this particular building is used mostly as a classroom and is unoccupied at night.  So the disturbance - and the transformed disturbance - should be small at night. The output disturbance estimated by the LD method is not shown since it is not comparable with the transformed input  disturbance identified by the \algoname\ method. 
	
	\textbf{\textit{Zone temperature fitting}}
	The temperature is predicted quite well by the identified plant and disturbance; see Fig.~\ref{fig:Tz_cmp_bldg}. The rms error is 0.3\degree C for the proposed method, and 0.1\degree C for the LD method.
	\begin{figure}[htb]
		\includegraphics[width=0.99\linewidth]{\figPATH/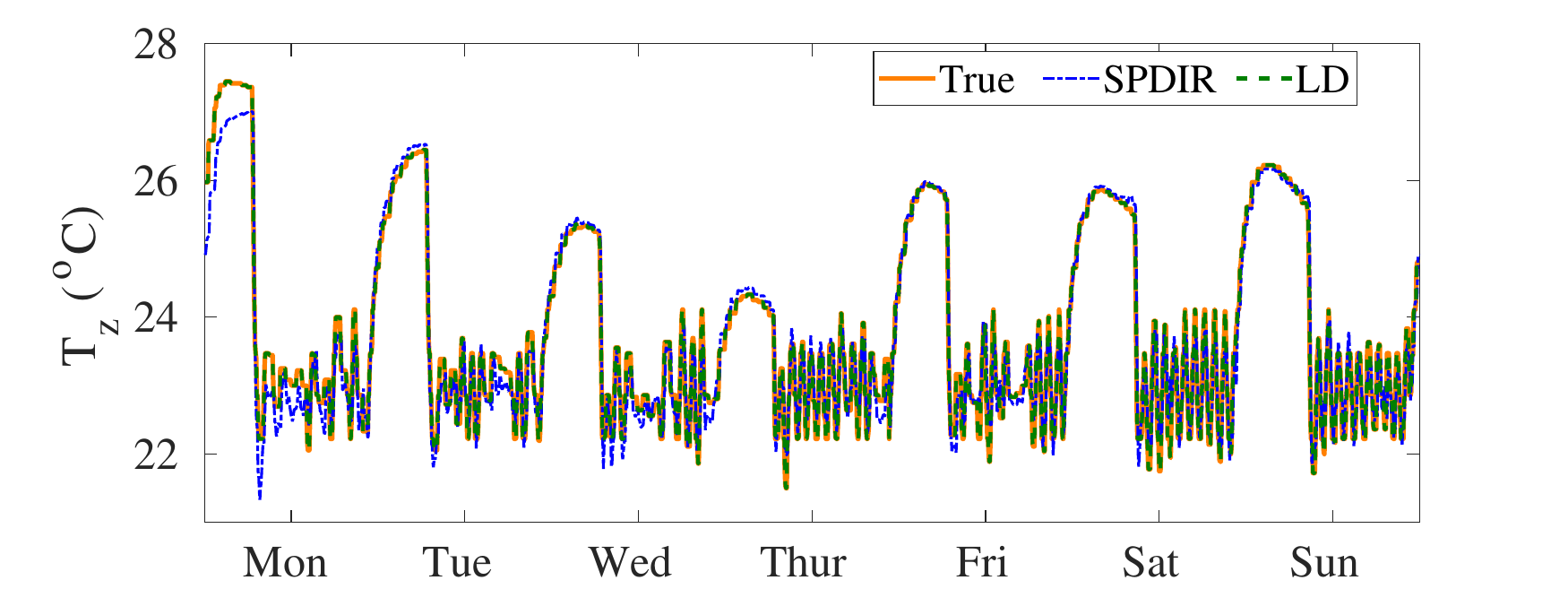}
		\centering
		\caption{Algorithm evaluation on building data: comparison of actual zone temperature and fitted zone temperatures.}
		\label{fig:Tz_cmp_bldg}
	\end{figure}
	
	\ifArxivVersion
	In addition, we also tried application of the Box-Jenkins method with an ARMAX (Autoregressive Moving Average with Explanatory Variable) model to identify an LTI plant driven by colored Gaussian disturbance~\cite{box2015time}. The proposed method outperforms the Box Jenkins method for experiments with both simulation and building data; see Sec.~\ref{sec:append_cmp_BJ} for the details.
	\else
	Though not reported here due to lack of space, we also tried application of the Box-Jenkins method with an ARMAX (Autoregressive Moving Average with Explanatory Variable) model to identify an LTI plant driven by colored Gaussian disturbance~\cite{box2015time}. The proposed method outperforms the Box Jenkins method for experiments with both simulation and building data; see~\cite{ZengSimultaneousArXiV:2017} for the details. 
	\fi
	
	\section{Conclusion}\label{sec:conclusion}
	The proposed method identifies a black box LTI model and a non-parametric (transformed) disturbance using
        $\ell_1$-regularization. In contrast to existing methods, it can be used as part of a self-learning control system without human supervision due to convexity and guarantees on stability and DC gains. Some work on using the method in this way for control are reported in~\cite{zeng2020autonomous}. There are many avenues for future work, including identification of the input disturbance rather than its transformed version, and analysis of the quality of data needed for the method to perform well. Modeling of multi-zone buildings is another area for extension. Preliminary work in this direction are reported in~\cite{ZengIdentification:TCST:2019}. 
	
	\bibliographystyle{plain}

	\ifArxivVersion
	\section{Appendix}\label{sec:appendix}
	\subsection{Constraints derivation} \label{sec:append_const}
	\textbf{\textit{Stability}}
	For discrete-time transfer function~\eqref{eq:dtmodel} to be stable, roots of $D(\zinv) = \zinvv(z^{2}-\theta_1z-\theta_2)$ need to lie within the unit circle, i.e., $|r_1|, |r_2| < 1$, where $r_1, r_2 = \frac{\theta_1\pm\sqrt{\Delta}}{2}$. Resulting discussion goes as follows depending on sign of $\Delta=\theta_1^2+4\theta_2$.
	\begin{enumerate}[\qquad\quad]
		\item [\emph{Case 1:}] $\Delta=\theta_1^2+4\theta_2\geq0$ $\Rightarrow$ $r_1$, $r_2$ $\in\R$.
		\begin{align}\label{eq:delta_sqrt}
		\begin{aligned}
		|r_1|,|r_2|<1
		\iff& 
		\begin{cases}
		-2-\theta_1<\sqrt{\Delta}<2-\theta_1\\
		-2+\theta_1<\sqrt{\Delta}<2+\theta_1
		\end{cases}.
		\end{aligned}
		\end{align}
		Notice that either $\theta_1\geq2$ or $\theta_1\leq-2$ will lead to $0\leq\sqrt{\Delta}<0$, which has no solution, we need
		\begin{align*}
		-2<\theta_1<2.
		\end{align*} 
		Take square of both sides of Eq~\eqref{eq:delta_sqrt} one would get:
		\begin{align*}
		\theta_2 + \theta_1 <1, \theta_2-\theta_1 < 1.
		\end{align*}
		\item [\emph{Case 2:}] $\Delta=\theta_1^2+4\theta_2<0$ $\Rightarrow$ $\bar{r}_1$=$r_2$.
		\begin{align*}
		\begin{aligned}
		|r_1|,|r_2|<1
		\iff&|r_1^2|=|r_1\bar{r}_1|=|r_1r_2|<1\\
		\iff&|-\theta_2|<1\\
		\iff&-1<\theta_2<1.
		\end{aligned}
		\end{align*}
	\end{enumerate}
	Combining the above two cases together, we have stability constraints as
	\begin{align*}
	\begin{aligned}
	-\theta_2& < 1, & \theta_2+\theta_1&< 1, & \theta_2-\theta_1& < 1.
	\end{aligned}
	\end{align*}  
	
	\textbf{\textit{Positive DC-gain}}
	For the discrete-time model~\eqref{eq:dtmodel}, the DC gains from inputs \qhvac, $T_{oa}$, \solirr, to the output $T_z$ are 
	\begin{align*}
	\begin{aligned}
	\frac{K(\zinv)^T}{D(\zinv)}|_{z=1} = 
	\frac{1}{1-\theta_1-\theta_2}\begin{bmatrix}
	\theta_3+\theta_4+\theta_5 \\
	\theta_6+\theta_7+\theta_8 \\
	\theta_9+\theta_{10}+\theta_{11} 
	\end{bmatrix}.
	\end{aligned}
	\end{align*}
	Using the previously established relation~\eqref{eq:stable_const2}:
	\begin{align*}
	\theta_2+\theta_1&< 1,
	\end{align*}
	the denominator $D(\zinv)|_{z=1}$ is positive. Therefore positive DC gains are equivalent to
	\begin{align*}
	\begin{aligned}
	\theta_3+\theta_4+\theta_5&> 0,
	\end{aligned}\\
	\begin{aligned}
	\theta_6+\theta_7+\theta_8&> 0,
	\end{aligned}\\
	\begin{aligned}
	\theta_9+\theta_{10}+\theta_{11}&> 0.
	\end{aligned}
	\end{align*}
	
	\textbf{\textit{Sign of parameters}}	
	If $t_s<min\{l,m,n\}$, where
	\begin{align*}
	l&\triangleq \frac{2C_wR_wR_z}{R_z+R_w},\notag \\ 
	m&\triangleq 2(R_zC_zR_wC_w)^{1/2},\notag\\
	n&\triangleq \frac{2}{3}\min(R_zC_z,R_zC_w,R_wC_w),
	\end{align*}
	the following derivations hold. Notice that in~\eqref{eq:theta_def}, signs of $\theta_i$'s depend only on numerators as they share a common positive denominator $D_0$, whose parameters are positive as shown in~\eqref{eq:def-D(s)}. Since $f_{12}$, $g_{22}>0$, from the last equation in~\eqref{eq:theta_def}, it follows that $\theta_6$, $\theta_7$, $\theta_8>0$. Meanwhile, we know that $t_s<l=-\frac{2}{f_{22}}$ by hypothesis and $f_{22}<0$; it follows that $-2-f_{22}t_s<0$. Given that $g_{11}$, $g_{13}$ are positive, it can be shown that $\theta_3$, $\theta_9<0$ while $\theta_4$, $\theta_5$, $\theta_{10}$, $\theta_{11}>0$.\\ 
	Similar analysis applies for $\theta_1$, $\theta_2$. Because 
	\begin{align*}
	0\leq t_s<\{2(R_zC_zR_wC_w)^{1/2}\},
	\end{align*} we have
	\begin{align*}
	t_s^2<4R_zC_zR_wC_w = \frac{4}{d_2} \iff 2d_2t_s^2<8.
	\end{align*}
	Thus from~\eqref{eq:theta_def}, $\theta_1>0$.
	Denote
	\begin{align*}
	\begin{aligned}
	b&\triangleq \frac{d_1}{d_2}=R_wC_w+R_wC_z+R_zC_z>0, \\
	c&\triangleq \frac{1}{d_2}=R_zC_zR_wC_w>0. 
	\end{aligned}
	\end{align*}
	From~\eqref{eq:theta_def}, $\theta_2<0$ is equivalent to
	\begin{align}\label{eq:theta2_parabola}
	\begin{aligned}
	d_2(t_s^2-\frac{2d_1}{d_2}t_s+\frac{4}{d_2})>0 \iff t_s^2-2bt_s+4c>0
	\end{aligned}
	\end{align}
	since $d_2>0$.
	Let $\Delta$ be the discriminant of the above quadratic, then
	\begin{align*}
	\begin{aligned}
	\frac{\Delta}{4} = & b^2-4c\\
	=&(R_wC_w)^2+(R_wC_z)^2+(R_zC_z)^2\\&-2R_wR_zC_wC_z+2R_wR_zC_z^2+2R_w^2C_wC_z\\
	>&(R_wC_w)^2+(R_wC_z)^2+(R_zC_z)^2\\&-2R_wR_zC_wC_z+2R_wR_zC_z^2-2R_w^2C_wC_z\\
	=&(R_wC_w-R_wC_z-R_zC_z)^2\\
	\geq & 0,
	\end{aligned}
	\end{align*}
	which means~\eqref{eq:theta2_parabola} has two distinct real roots
	$r_1,r_2$. Here we will show \eqref{eq:theta2_parabola} holds by showing $t_s$ lies on the left hand side of the smaller root, denoted as $r_1$. We have
	\begin{align*}
	\begin{aligned}
	r_1 = & b-\sqrt{b^2-4c}=b-\sqrt{\bigg(b-\frac{2c}{b}\bigg)^2-\frac{4c^2}{b^2}}\\
	\geq &  b-\sqrt{\bigg(b-\frac{2c}{b}\bigg)^2}  = b-\sqrt{\bigg(\frac{b^2-2c}{b}\bigg)^2}.
	\end{aligned}
	\end{align*}
	Since $c>0$, we have $b^2-2c>b^2-4c>0$, and
	\begin{align*}
	\begin{aligned}
	r_1 \geq & b-\bigg(\frac{b^2-2c}{b}\bigg)=\frac{2c}{b}=  \frac{2R_zC_zR_wC_w}{R_wC_w+R_wC_z+R_zC_z} \\
	\geq &\frac{2R_zC_zR_wC_w}{3\max(R_wC_w,R_wC_z,R_zC_z)}\\
	=&\frac{2}{3}\min(R_zC_z,R_zC_w,R_wC_w)\\
	>& t_s,
	\end{aligned}
	\end{align*}
	where the final inequality is due to the initial hypothesis: $t_s<n$. Hence~\eqref{eq:theta2_parabola} holds, or equivalently, $\theta_2<0$.
	Therefore we have proved constraints on the signs of parameters as stated in~\eqref{eq:thetasigns}.

	\subsection{Comparison with the Lumped disturbance method without parameter tunning}\label{sec:append_cmp_LD}
	Provided below are the results we obtained to show that, the LD (lumped disturbance) method performs less satisfying without carefully tuned parameters. Both the proposed and the LD methods are implemented using the the same simulation dataset from Section~\ref{sec:evaluation}, except that the disturbance is designed to change from the blue signal (test 1) into the orange signal (test 2) as shown below in Figure~\ref{fig:dist_true}. The difference between these two disturbance signals has mean $0$, with $2.5$ kW standard deviation.
	\begin{figure}[h]
		\includegraphics[width=0.99\linewidth]{\figPATH/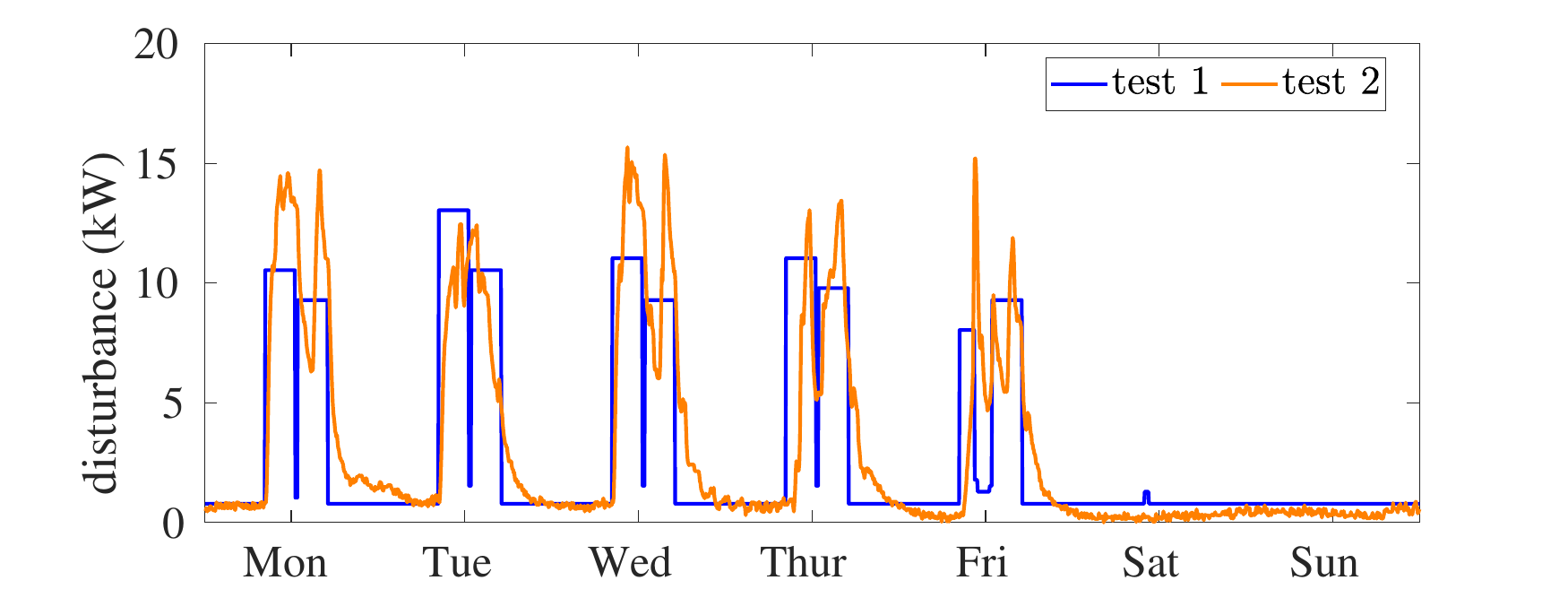}
		\centering
		\caption{Disturbance signals used in the simulation.}
		\label{fig:dist_true}
	\end{figure}
	
	Note for the LD method, an initial guess for the plant parameters need to be provided. These initial parameters are carefully picked in test 1, and are kept as the same in test 2. Results indicate that the zone temperature prediction accuracy are much worse, when the disturbance changes but the initial guesses for the parameters remain the same: the rms value of the prediction error doubles from 1.2\degree C to 2.4\degree C; see Figure~\ref{fig:Tz_cmp_LD}.
	\begin{figure}[h]
		\includegraphics[width=0.99\linewidth]{\figPATH/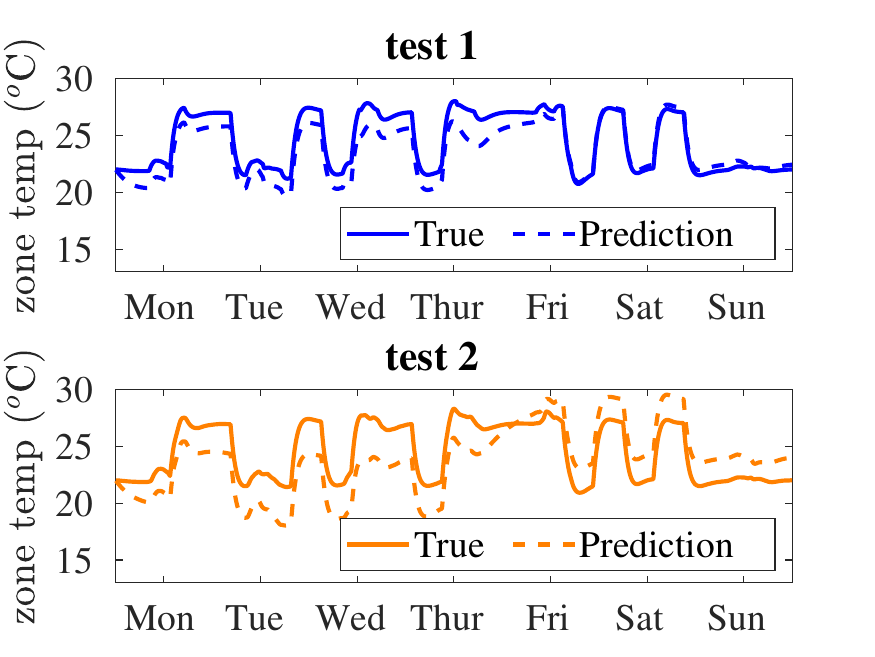}
		\centering
		\caption{Evaluating the LD method with simulation data: predicted and actual zone temperature.}
		\label{fig:Tz_cmp_LD}
	\end{figure}
	
	In contrast, the same experiments are repeated using the proposed method without retuning $\lambda$. The rms value of the prediction error changes slightly from 1.0 \degree C to 1.2\degree C, despite the changes in the disturbance; see Figure~\ref{fig:Tz_cmp_SPDIR}.
	\begin{figure}[h]
		\includegraphics[width=0.99\linewidth]{\figPATH/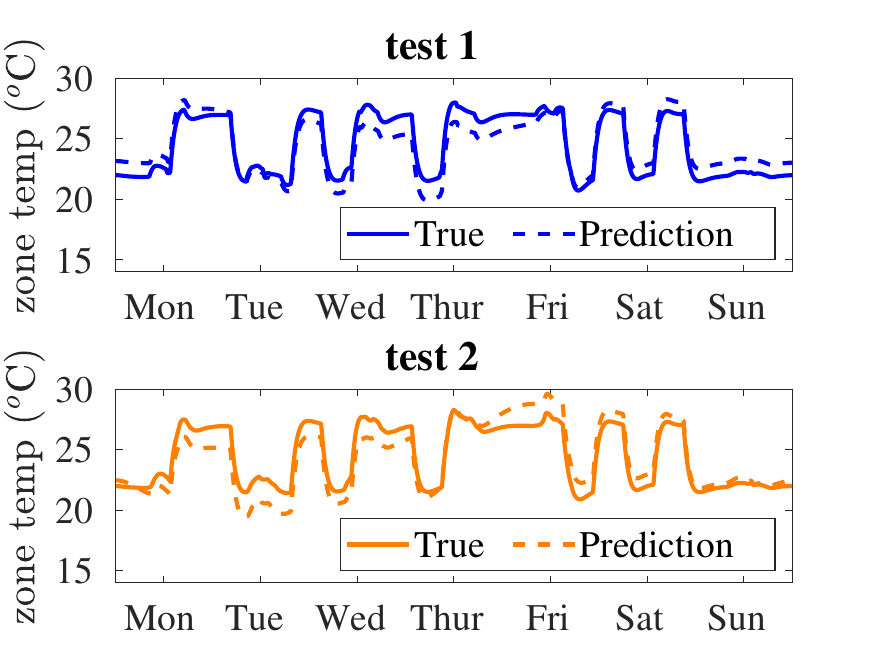}
		\centering
		\caption{Evaluating the SPDIR method with simulation data: predicted and actual zone temperature.}
		\label{fig:Tz_cmp_SPDIR}
	\end{figure}

	\subsection{Comparison with Box Jenkins method}\label{sec:append_cmp_BJ}
	The ARMAX model with order $(p,r,q)$ is:
	\begin{align*}
	\begin{aligned}
	y_t = &\phi_1y_{t-1}+\dots+\phi_py_{t-p}+
	\beta^T_0u_t+\dots + \beta^T_ru_{t-r} + \\
	& e_t+\theta_1e_{t-1}+\dots+\theta_qe_{t-q},\\
	& t=p+1,\dots,N,
	\end{aligned}
	\end{align*}
	in which $u\in \R^3$ has 3 input variables. Here $e_t$'s are white Gaussian and are assumed to be statistically independent of the input $u$. The model orders $(p,q)$ are first identified based on observation of the impulse response of the system. Then $r$ is identified based on inspections of the ACF (Auto-Correlation Function) and PACF (Partial Auto-Correlation Function) of the time series $y$. The model parameters are estimated using nonlinear least-squares.
	
	Below we provide comparisons of our SPDIR method with the Box Jenkins method, for experiments with simulation data and real building data respectively. 
	
	\textbf{\textit{Performance of Box Jenkins method applied to the ARMAX model using simulation data}}\\
	The Box Jenkins method is implemented with a ARMAX model using the same simulation data set from our paper. The model order $(p,r,q)$ are identified by using the methods provided in ref [1] (in the revised manuscript) to be $(2,2,2)$. For brevity, following are two excerpts from the results, one for disturbance estimate and the other for zone temperature prediction.
	\begin{itemize}
		\item Figure~\ref{fig:distEsti_cmp} shows the true values of the linear transformation of the disturbance, and the corresponding estimations from the SPDIR method and the BJ method. Results indicate that the proposed method estimates the disturbance more accurately.
		\begin{figure}[h]
			\includegraphics[width=0.99\linewidth]{\figPATH/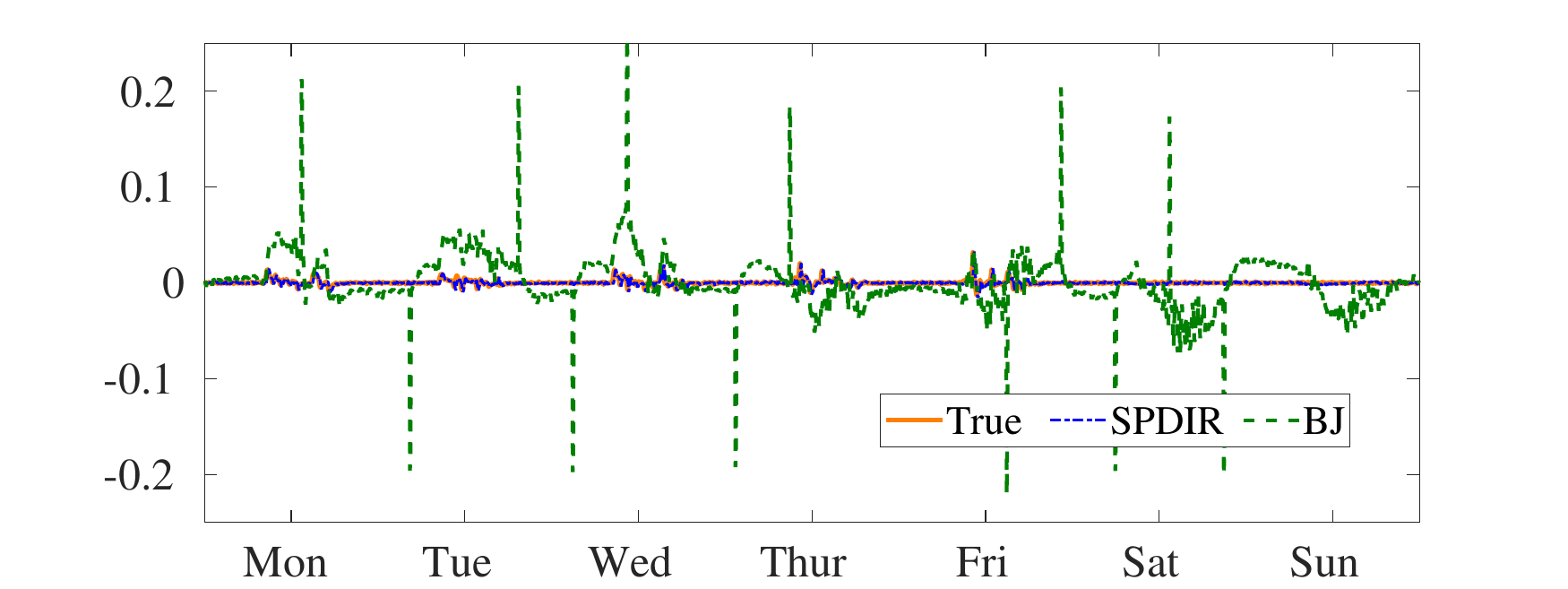}
			\centering
			\caption{Algorithm evaluation on simulation data: estimated and actual lumped disturbance.}
			\label{fig:distEsti_cmp}
		\end{figure}
		\item The plant identified with data from one week is used to predict temperatures in another week. The rms value of the prediction error of zone temperature is 2.3\degree C for the BJ method. The proposed SPDIR method predict the temperature more accurately, with rms error of 1.2\degree C; see Figure~\ref{fig:Tz_cmp}. 
		\begin{figure}[h]
			\includegraphics[width=0.99\linewidth]{\figPATH/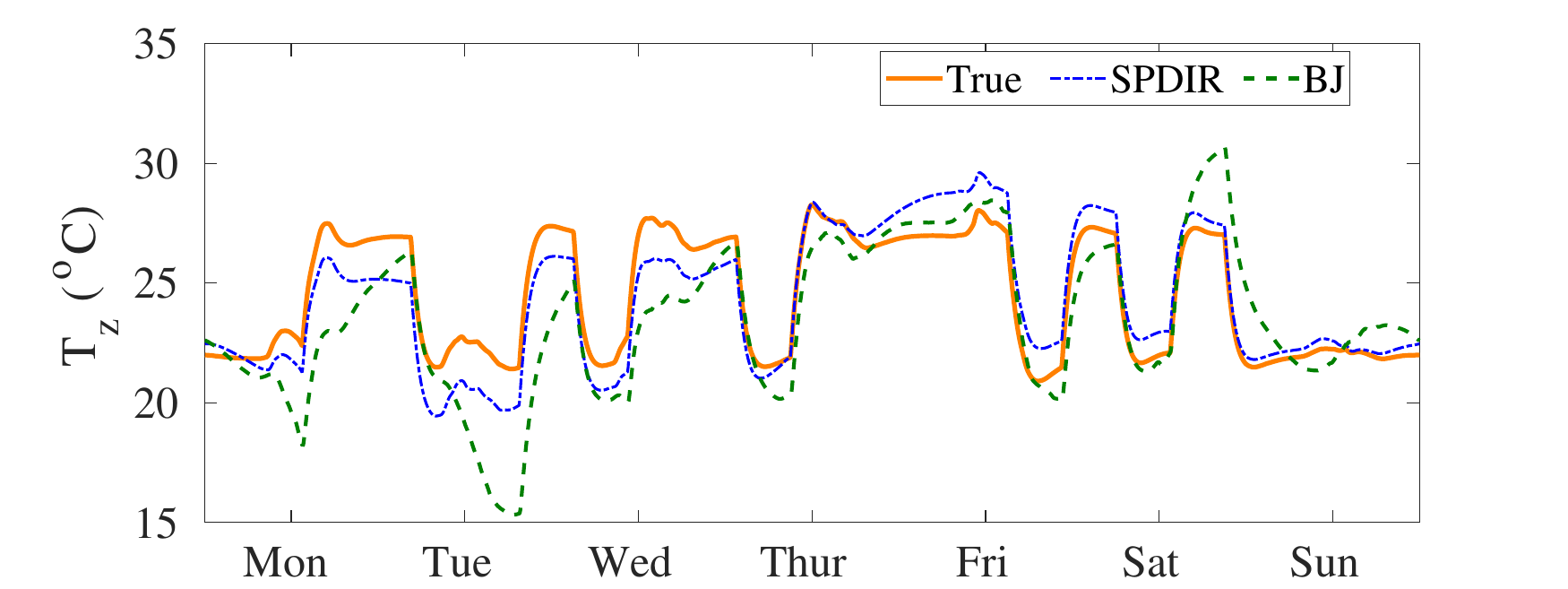}		
			\centering
			\caption{Algorithm evaluation on simulation data: predicted and actual zone temperature (validation dataset).}
			\label{fig:Tz_cmp}
		\end{figure}	
	\end{itemize}
	
	\textbf{\textit{Performance of Box Jenkins method applied to the ARMAX model using building data}}\\
	We now apply the methods to the building dataset from our paper. Since there is no ground truth to compare with, we only provide the zone temperature fitting results below.
	
	The rms value of the prediction error of zone temperature is 0.3\degree C for the proposed method, and is 4.5\degree C for the BJ method; see Figure~\ref{fig:Tr_bldg_prop_BJ}. 
	\begin{figure}[h]
		\includegraphics[width=0.99\linewidth]{\figPATH/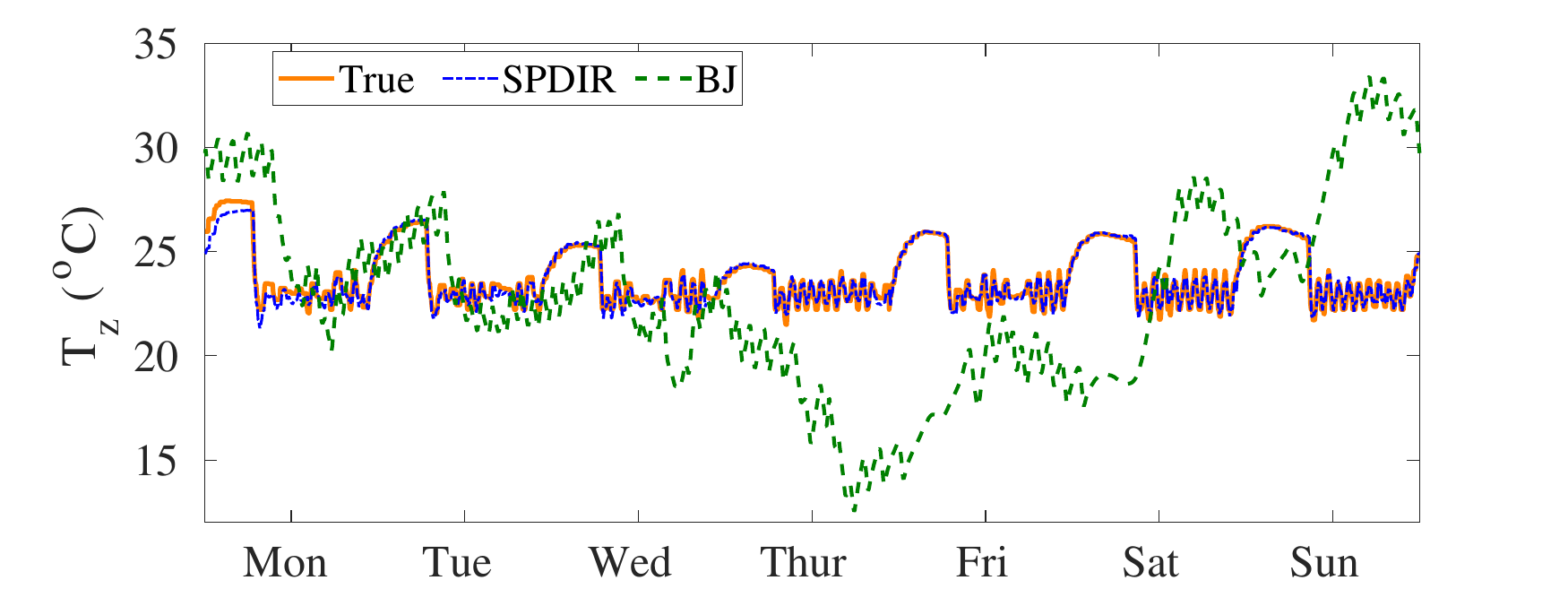}
		\centering
		\caption{Algorithm evaluation on building data: predicted and actual zone temperature.}
		\label{fig:Tr_bldg_prop_BJ}
	\end{figure}

	\fi

\end{document}